\documentclass{article}
\usepackage[utf8]{inputenc}
\usepackage[margin=2cm]{geometry}
\usepackage{enumerate}
\usepackage{amsmath,amsthm,amsfonts,amssymb, mathtools,mathrsfs, xcolor}
\usepackage{bm}
\usepackage{bbm}
\usepackage{hyperref}
\usepackage[style=numeric-comp, maxnames=4,
    minnames=4]{biblatex}
\usepackage{todonotes}
\usepackage{centernot}

\theoremstyle{plain}
\newtheorem{theorem}{Theorem}
\newtheorem{lemma}{Lemma}
\newtheorem{definition}{Definition}
\newtheorem{proposition}{Proposition}

\theoremstyle{definition}
\newtheorem{remark}{Remark}

\newcommand{\en}[1]{\omega_{\textbf{#1}}}
\newcommand{\di}[0]{\text{d}}
\newcommand{\massh}[0]{H_m}

\newcommand{\real}[0]{\operatorname{Re}}
\newcommand{\immag}[0]{\operatorname{Im}}

\newcommand{\RealSub}[1]{\mathscr{K}_{#1}}
\newcommand{\Fock}[1]{\mathscr{F}^s(#1)}
\newcommand{\Hil}[0]{\mathscr{H}}
\newcommand{\RealHil}[0]{\mathcal{H}}
\newcommand{\ComplexProduct}[2]{\left\langle#1,#2\right\rangle}
\newcommand{\RealProduct}[2]{\left(#1,#2\right)}
\newcommand{\Sympl}[2]{\sigma\left(#1,#2\right)}
\newcommand{\ComplexStr}[0]{\beta}

\newcommand{\Ocal}{\mathcal{O}}
\newcommand{\UOcal}[0]{ \mathcal{U}_\Ocal}
\newcommand{\VOcal}[0]{ \mathcal{V}_\Ocal}

\newcommand{\Cauchy}[0]{\Sigma}
\newcommand{\Mink}[0]{\mathbb{M}}
\newcommand{\Cdiamond}[0]{C}
\newcommand{\subMink}[0]{\textbf{B}}
\newcommand{\complexconj}[0]{\operatorname{C}}
\newcommand{\implicitop}[0]{\operatorname{Z}^\beta}
\newcommand{\opham}[0]{h}
\newcommand{\boundedop}[1]{\mathcal{B}(#1)}

\newcommand{\operatornorm}[1]{\|#1\|_{\operatorname{op}}}
\newcommand{\Fphi}[0]{\mathscr{F}_{\varphi}}
\newcommand{\Fpi}[0]{\mathscr{F}_{\pi}}
\newcommand{\Abweyl}[0]{\mathscr{W}}
\newcommand{\spec}[1]{\rho(#1)}
\newcommand{\spanc}[0]{\operatorname{span}_\mathbb{C}}

\renewcommand{\thefootnote}{\fnsymbol{footnote}}

\begin{document}
\par 
\bigskip 
\LARGE 
\noindent 
\textbf{Haag Duality in the Thermal Sector} 
\bigskip \bigskip
\par 
\rm 
\normalsize 
 
\large
\noindent 
{\bf Stefano Galanda\footnote{Corresponding author}$^{1,a}$}, {\bf Leonardo Sangaletti$^{2,b}$}\\
\par
\small

\noindent$^1$Department of Mathematics, University of York - 
York YO10 5GH, United Kingdom. \smallskip

\noindent$^2$Dipartimento di Fisica, Universit\`a di Genova - Via Dodecaneso, 33, I-16146 Genova, Italy. \smallskip
\smallskip

\noindent E-mail: 
$^a$stefano.galanda@york.ac.uk, 
$^b$leonardo.sangaletti@edu.unige.it

\bigskip

\noindent 
\small 
\renewcommand{\thefootnote}{\arabic{footnote}}
\setcounter{footnote}{0}
{\bf Abstract.}  We prove that the net of localised von Neumann algebras associated with a real scalar field propagating on Minkowski spacetime, in the KMS representation, satisfies a generalised version of Haag duality. Our proof combines ideas from existing arguments for the ground-state representation with purification techniques.

\bigskip


\section{Introduction}
Haag duality is a central structural property of algebraic quantum field theory (AQFT) as formulated by the Haag-Kastler axioms~\cite{HaaKas:algebraic_approach}, expressing a precise correspondence between the orthomodular lattices of causally complete spacetime regions and the nets of local algebras. In other words, it establishes an exact equivalence between spacetime localization and algebraic commutation relations. In its standard formulation, Haag duality states that the von Neumann algebra, in an irreducible representation, localised in a \textit{sufficiently regular} open region coincides with the commutant of the algebra associated with its causal complement. This condition ensures that the local algebras are maximal, i.e.~no additional observables can be added without violating the condition of locality.\\

Beyond its conceptual significance, Haag duality plays a crucial technical role in the understanding and proof of several results in quantum field theory. Most notably, it is one of the essential properties needed to carry out the analysis of superselection sectors as carried out by Doplicher–Haag–Roberts~\cite{DHRa, DHRb, Dop:superselection_1, DopLon:superselection_2}. Moreover, it also underpins results on the modular structure of local algebras, such as the Bisognano–Wichmann theorem~\cite{BisognanoWichmann:1976}, thereby linking locality with spacetime symmetries.\\

Duality for the local net of von Neumann algebras was first rigorously established for free massive fields in the ground state representation on Minkowski spacetime. The first proof was given by Araki~\cite{ArakiHD1, ArakiHD2} who showed it for a free real scalar field localised in a double cone (causal diamond). Subsequently, alternative proofs of the same statement were given by Eckmann and Osterwalder using Tomita–Takesaki modular theory for standard subspaces~\cite{Osterwalder1973, EckmannOsterwalder1973} and by Rieffel~\cite{Rieffel1974}, see also~\cite{Garbarz_2022} for a review. This result was further extended to bosonic fields with spin by Dell'Antonio in~\cite{DellAntonio1968}. In the analogous case of free fermionic fields, a twisted version of Haag duality (due to canonical anticommutation relations) was proven in the Appendix of~\cite{DHRa}. Further proofs have been given for conformal field theories~\cite{BuchholzSchulz-Mirbach:1990, BrunettiGuidoLongo1993} and no-go results have been provided for gauge theories due to the existence of global charges~\cite{LeylandRobertsTestard1978}. Finally, these results can be generalised to globally hyperbolic spacetimes, in the representation induced by Hadamard states as proven and pioneered in the construction presented in \cite{Verch:1993, Verch_1997}.\\

Beyond continuum fields, Haag duality in the ground state representation has found important applications in lattice models of quantum physics, where it serves as a rigorous bridge between AQFT ideas and condensed matter systems. Existing proofs treat the case of one-dimensional spin systems~\cite{KeylMatsuiSchlingemannWerner2006} and two-dimensional lattice models such as the toric code~\cite{Naaijkens2012} and more general Kitaev quantum double models for cone-like regions of localisation~\cite{FiedlerNaaijkens2015, OgataPerezGarciaRuizDeAlarcon2025}. These results enabled the classification of anyonic superselection sectors via a Doplicher–Haag–Roberts–type analysis on the infinite lattice~\cite{Ogata2020Z2Index, Ogata2021H3Index, Ogata2022FermionSPT}. An approximate version of Haag duality~\cite{Ogata2022BraidedCategories} has also been used to prove the quantisation of the Hall conductance for an infinite plane~\cite{BachmannCorbelliFraasOgata2025}. Finally, let us mention that Haag duality has found an important application also in the recent program of formulating quantum information using algebraic methods~\cite{Luijk2025EntanglementCMP, LuijkPurificationHaag}.\\

However, despite its evident relevance across various areas of mathematical physics, a formulation of this duality in representations other than those induced by pure quasi-free states (such as the ground state) appears, to the best of our knowledge, to be missing. The duality relation for reducible representations $\pi^\beta$ induced by thermal equilibrium states has recently gained some interest~\cite{Bostelmann:2025bfr}. One would expect that the usual Haag-Duality relation should be modified in a natural way to take into account the reducibility of the representation (see~\cite[Def.~$37$]{FewsterRejzner}). The aim of this paper is to prove this relation. Consider the representation of the quasi-local $\operatorname{C}^*$-algebra induced by the ground state, denoted by $\mathcal{A}$. Since the representation is irreducible, the corresponding von Neumann algebra $\mathcal{A}''$ coincides with $\mathcal{B}(\Hil)$, namely the von Neumann algebra of bounded linear operators acting on a certain Hilbert space $\Hil$. Since Haag duality holds in this representation, we have that $\mathcal{M}(\Ocal)'=\mathcal{M}(\Ocal')$, with $\Ocal$ being a sufficiently regular bounded open region of spacetime and $\mathcal{M}(\Ocal)$ the corresponding local von Neumann algebra. Consider now a second copy of the quasi-local von Neumann algebra $\mathcal{B}(\Hil)$ and its (von Neumann algebra) tensor product $\mathcal{B}(\Hil)\otimes\mathcal{B}(\Hil)$. The commutant of the local algebra, identified with $\mathcal{M}(\Ocal)\otimes\mathbbm{1}$ in the doubled quasi-local algebra coincides with $\mathcal{M}(\Ocal')\otimes\mathcal{B}(\Hil)$~\cite[Lem.~$5.4$]{ArakiHD1}. If the doubled ground representation was unitarily equivalent to the thermal one, we could derive our result from the previous analysis. However, as it is well known, the two representations are not unitarily equivalent when the system considered is infinitely extended in space. Intuitively, this inequivalence is understood as a consequence of the infinite number of particles that distinguish the ground representation and a representation at finite temperature, in which the system is described by an infinitely extended ``thermal bath''. More precisely, it follows from the fact that the ground representation of the quasi-local $\operatorname{C}^*$-algebra generates a von Neumann algebra of Type I (coincides with $\mathcal{B}(\Hil)$), while in the finite temperature case with a von Neumann algebra of Type III. In order to circumvent the inequivalence between the two representations we first consider the problem at the one particle level. At this level, the ``infinite number of particles'' that distinguish the ground and the thermal representation play no role. Therefore, the two representations are related by a well defined Bogoliubov transformation. Once the result at the one particle level is obtained, we lift it at the level of the represented algebras by second quantization techniques~\cite{ArakiHD1,ArakiHD2,EckmannOsterwalder1973}.\\

The paper is organised as follows. In Section~\ref{sec: Setup} we recall the definition of the net of local von Neumann algebras associated with a free real scalar field in the ground and KMS representations, and we state the main result of our paper. In Section~\ref{sec: General}, following the original works of Araki and Eckmann–Osterwalder and combining them with the technique of purifying the representation induced by a KMS state (see~\cite{Umezawa:1982nv} for a review), we provide a general proof of duality in the KMS representation. In particular, the reducibility of the original representation gives rise to an additional non-trivial component in the commutant of a local algebra, beyond the usual geometric one. At the one-particle level, the real subspaces labelling the local algebras are embedded into a larger ambient space via a Bogoljubov transformation. It is therefore necessary to analyse how the relative position of these real subspaces is affected by the doubling procedure used to purify the thermal equilibrium state. Finally, in Section~\ref{sec: ExMink} we apply the general results obtained in the previous section to the specific example of a real massive scalar field on Minkowski spacetime. The Appendices collect, due to their relevance for our proof, the statement of the main theorem in~\cite{EckmannOsterwalder1973}, two technical lemmata and a proof of the so called \textit{pre-cyclicity property} of local real subspaces of a scalar field in a thermal representation.

\subsection{Notation}
For $f_1,f_2 \in S(\mathbb{R}^4,\mathbb{R})$ Schwartz functions, we denote their Fourier and inverse Fourier transform on Minkowski spacetime $\mathbb{M}$ respectively by
\begin{equation*}
    \hat{f}_1(p) \coloneqq \int f_1(x) e^{-ipx} \di^4 x \, , \qquad \check{f}_2(x) \coloneqq \frac{1}{(2\pi)^{4}} \int f_2(p) e^{ipx} \di^4 p
\end{equation*}
where in the exponential $px=p_\mu x^\mu$ is the product of the corresponding four-vectors in the mostly plus convention for the metric.\\
For $g_1,g_2 \in S(\mathbb{R}^3,\mathbb{R})$ we denote the spatial Fourier and spatial inverse Fourier transform by
\begin{equation*}
    \mathcal{F}(g_1)(\mathbf{p}) = \int g_1(\mathbf{x}) e^{-i\mathbf{p} \mathbf{x}}\di^3 \mathbf{x} \, , \qquad \mathcal{F}^{-1}(g_2)(\mathbf{x}) = \frac{1}{(2 \pi)^3}\int g_2(\mathbf{p}) e^{i\mathbf{p} \mathbf{x}}\di^3 \mathbf{p}.
\end{equation*}

\section{Setup and results}\label{sec: Setup}

\subsection{Real scalar field}\label{Sec: RealScalarField}
The model considered in this paper is a free scalar field on Minkowski spacetime. In the $\operatorname{C}^*$-algebraic setting, canonical commutation relations (CCR) are encoded by the algebraic relations among the generators of the so called \textit{Weyl algebra}
\begin{definition}
Let $(\mathcal{D},\sigma)$ be a real symplectic space consisting of a real linear vector space $\mathcal{D}$ and a non-degenerate symplectic form $\sigma: \mathcal{D} \times \mathcal{D} \to \mathbb{R}$. The Weyl $\operatorname{C}^*$-algebra associated with the real symplectic space $(\mathcal{D},\sigma)$, denoted $\mathcal{A}(\mathcal{D},\sigma)$, is the $\operatorname{C}^*$-algebra generated by symbols $\Abweyl(f),\;f \in \mathcal{D}$ satisfying the following relations for any $f,g\in\mathcal{D}$
\begin{align*}
    &\Abweyl(f)^* = \Abweyl(-f)\\
    &\Abweyl(f) \Abweyl(g) = e^{-\frac{i\sigma(f,g)}{2}}\Abweyl(f+g)\\
    &\Abweyl(0) = \mathrm{id}_{\mathcal{A}(\mathcal{D},\sigma)}
\end{align*}
and endowed with a $\operatorname{C}^*$-norm (see~\cite{bratteli1979operator}).
\end{definition}
It is well known that the Weyl $\operatorname{C}^*$-algebra exists and is unique up to $^*$-isomorphisms~\cite{Slawny:1972iq}. As we are interested in studying a free real massive scalar field on Minkowski spacetime $\mathbb{M}$, we consider the bilinear antisymmetric degenerate form on $ C^{\infty}_c(\mathbb{M},\mathbb{R})$
\begin{equation*}
    \sigma(f,g) \coloneqq \operatorname{E}(f,g)
\end{equation*}
where $\operatorname{E}: C^{\infty}_c(\mathbb{M},\mathbb{R}) \times C^{\infty}_c(\mathbb{M},\mathbb{R}) \to \mathbb{R}$ is the unique causal propagator associated with the Klein-Gordon equation with mass $m > 0$. The pair $(C^{\infty}_c(\mathbb{M},\mathbb{R})/\ker (\operatorname{E}),\operatorname{E})$ defines a symplectic space and we denote by $\mathcal{A}(\mathbb{M})=\mathcal{A}(C^\infty_c(\mathbb{M},\mathbb{R})/\ker \operatorname{E},\operatorname{E})$ the corresponding Weyl algebra. \\

Subordinated to the above quasi-local definition, we introduce a local net of $\operatorname{C}^*$-algebras~\cite{HaaKas:algebraic_approach}. This net associates to any open, bounded region of spacetime $\Ocal\in\Mink$ a $\operatorname{C}^*$-algebra $\mathcal{A}(\Ocal)$ via
\begin{equation*}
    \Ocal\mapsto\mathcal{A}(\Ocal)=\mathcal{A}(C^\infty_c(\Ocal,\mathbb{R})/\ker \operatorname{E},\operatorname{E}),
\end{equation*}
where $C^\infty_c(\Ocal,\mathbb{R})$ denotes the linear space of real valued test functions whose support is contained in $\Ocal$. The causal properties of $\operatorname{E}$ imply that this net is local, i.e $[\mathcal{A}(\Ocal_1), \mathcal{A}(\Ocal_2)]=\{0\}$ if $\Ocal_1,\Ocal_2$ are causally separated.
The quasi-local Weyl algebra $\mathcal{A}(\mathbb{M})$ coincides with $\operatorname{C}^*$-inductive limit of the net
\begin{equation*}
    \mathcal{A}(\mathbb{M}) = \overline{\bigcup_{\mathcal{O} \subset \mathbb{M}} \mathcal{A}(\mathcal{O})}^{\operatorname{C}^*},
\end{equation*}
By construction, this net satisfies isotony, i.e. $\Ocal_1 \subset \Ocal_2$ implies $\mathcal{A}(\Ocal_1) \subset \mathcal{A}(\Ocal_2)$. Furthermore, the invariance of the symplectic form under the action of the identity connected component of the Poincaré group $\mathscr{G}$ implies the covariance of the net; namely it exists a group of automorphisms $\alpha_\mathfrak{g}$ of the net such that for any $\mathfrak{g}\in\mathscr{G}$ the geometric action
\begin{equation*}
    \alpha_\mathfrak{g}\mathcal{A}(\mathcal{O})=\mathcal{A}(\mathfrak{g}(\mathcal{O}))
\end{equation*}
holds.\\

A state over an abstract $\operatorname{C}^*$-algebra, such as the Weyl algebra, is a positive, linear and normalised functional $\omega: \mathcal{A} \to \mathbb{C}$. Once a state on the $\operatorname{C}^*$-algebra $\mathcal{A}$ is chosen, a corresponding representation of its element as bounded operators on a Hilbert space is obtained via the GNS construction. To any state $\omega$ on $\mathcal{A}$, we associate a triple $(\mathscr{H}_{\omega}, \pi_{\omega}, \Omega_{\omega})$ where $\mathscr{H}_\omega$ is a complex Hilbert space, $\pi_{\omega}: \mathcal{A} \to \boundedop{\mathscr{H}_{\omega}}$ is a $^*$-homomorphism and $\Omega_{\omega} \in \mathscr{H}_{\omega}$ a cyclic vector that implements the state $\omega$
\begin{equation*}
    \omega(A)=(\Omega_\omega,\pi_\omega(A)\Omega_\omega)_{\mathscr{H}_\omega}, \quad\forall A\in\mathcal{A}.
\end{equation*}
This triple is unique up to unitarily equivalence.\\

If a state $\omega$ is quasi-free the corresponding GNS representation of the Weyl algebra $\mathcal{A}(\mathcal{D},\sigma)$ is of Fock type~\cite[Prop.~3.1]{KW91}. Usually, the free \textit{ground state} $\omega^{\infty}$ and the free \textit{KMS states} $\omega^{\beta}$ at inverse temperature $0<\beta<\infty$ defined on the quasi-local Weyl algebra with respect to the free dynamics are assumed to be quasi-free. We recall that a \textit{dynamics} on a $\operatorname{C}^*$-algebra $\mathcal{A}$ is a one-parameter family of $^*$-automorphism $\alpha_t,\;t\in\mathbb{R}$ of $\mathcal{A}$. Specifically, the \textit{free dynamics} on $\mathcal{A}(\Mink)$ acts on the generators as
\begin{equation*}
    \alpha_t \Abweyl(f) = \Abweyl(\mathscr{T}_t(f)) \qquad \forall t \in \mathbb{R}
\end{equation*}
where $\mathscr{T}_t$ is the one-parameter family of symplectomorphisms of $(C^{\infty}_c(\mathbb{M}/\ker \operatorname{E},\mathbb{R}), \operatorname{E})$ that implements the time shift isometries of $\Mink$, namely the Killing flow associated with the global timelike killing vector $\partial_t$
\begin{equation*}
    \mathscr{T}_t(f)(x^0,x^1,x^2,x^3)=f(x^0-t,x^1,x^2,x^3).
\end{equation*}
\\

The GNS representation of $\mathcal{A}(\Mink)$ induced by $\omega^\beta$ can be obtained by following~\cite{Kay:1985yx,Kay:purification} (see also~\cite{ArakiWoods}). First of all, we recall the definition of \textit{ground one-particle structure} over a symplectic space $(\mathcal{D},\sigma)$ equipped with a one-parameter family of linear symplectomorphisms $\mathscr{T}(t)$
\begin{definition}[\cite{Kay:1985yx} \textbf{Definition \boldmath$1\textbf{a}$\unboldmath}]\label{def: groundrep}
A ground one-particle structure $(K^{\infty}, \mathscr{H}^\infty, e^{-it h})$ over $(\mathcal{D},\sigma, \mathscr{T}(t))$ consists of a complex Hilbert space $\mathscr{H}^\infty$, a map $K^{\infty}: \mathcal{D} \to \mathscr{H}^\infty$ and a strongly continuous unitary group $e^{-it h}$ on $\mathscr{H}^\infty$ such that
\begin{enumerate}
    \item The image of $K^{\infty}$ is a real-linear subspace of $\mathscr{H}^\infty$ and is such that for any $f_1,f_2 \in \mathcal{D}$:
    \begin{equation*}
        2 \mathrm{Im}\langle K^{\infty} f_1, K^{\infty} f_2 \rangle_{\mathscr{H}^\infty}  = \sigma(f_1,f_2).
    \end{equation*}
    \item $K^{\infty} \mathcal{D}$ is dense in $\mathscr{H}^\infty$.
    \item $K^{\infty} \mathscr{T}(t) = e^{-it h} K^{\infty}$ on $\mathcal{D}$, where $h$ is a, possibly unbounded, strictly positive linear operator on $\mathscr{H}^\infty$.\footnote{By strictly positive we mean that its spectrum is contained in $(0,+\infty)$. More generally, an analogous definition can be given requiring $h$ to be only positive, self-adjoint and with no zero eigenvalue; however, for the purpose of our proof, the request of strict positivity turns out to be necessary. Therefore, we adopt it from the very beginning.}
\end{enumerate}
\end{definition}
The ground one-particle structure is unique up to unitary equivalence~\cite[Thm.~1a.]{Kay:1985yx}.
It follows that the symmetric Fock space constructed on the ground one-particle Hilbert space $\mathscr{H}^\infty$ coincides with the Hilbert space of the GNS representation of $\mathcal{A}(\mathcal{D},\sigma)$ induced by the quasi-free ground state $\omega^{\infty}$
\begin{equation}\label{Eq: ground_fock}
    \mathscr{H}_{\omega^{\infty}} = \bigoplus_{n=0}^{\infty} (\mathscr{H}^\infty)^{\otimes_s n}.
\end{equation}
On it, the ground state $\omega^\infty$ is implemented by the cyclic vacuum vector $\Omega_\infty\in(\mathscr{H}^\infty)^{\otimes_s0}=\mathbb{C}$.\\
We now recall the definition of \textit{KMS one-particle structure}
\begin{definition}[\cite{Kay:1985yx} \textbf{Definition \boldmath$1\textbf{b}$\unboldmath}]\label{Def: KMS-1ps}
A KMS one-particle structure $(K^{\beta}, \mathscr{H}^\beta, e^{-it \Tilde{h}})$ over $(\mathcal{D},\sigma, \mathscr{T}(t))$, consists of a complex Hilbert space $\mathscr{H}^\beta$, a map $K^{\beta}: \mathcal{D} \to \mathscr{H}^{\beta}$ and a strongly continuous unitary group $e^{-it \Tilde{h}}$ on $\mathscr{H}^\beta$ such that:
\begin{enumerate}
    \item The image of $K^{\beta}$ is a real-linear subspace of $\Hil^\beta$ and is such that for any $f_1,f_2 \in \mathcal{D}$
    \begin{equation*}
        2 \mathrm{Im}\langle K^{\beta} f_1, K^{\beta} f_2 \rangle_{\mathscr{H}^\beta}  = \sigma(f_1,f_2).
    \end{equation*}
    \item $K^{\beta} \mathcal{D} + i K^{\beta} \mathcal{D}$ is dense in $\mathscr{H}^\beta$.
    \item $K^{\beta} \mathscr{T}(t) = e^{-it\Tilde{h}} K^{\beta}$ on $\mathcal{D}$, with $\Tilde{h}$ having trivial kernel and satisfying the one-particle KMS condition 
    \begin{equation*}
        \langle e^{-it \Tilde{h}} K^\beta f_1,K^\beta f_2 \rangle_{\mathscr{H}^\beta} = \langle  e^{-\beta \Tilde{h}/2} K^\beta f_1, e^{-it \Tilde{h}} e^{-\beta \Tilde{h}/2}K^\beta f_2 \rangle_{\mathscr{H}^\beta} \qquad \forall t \in \mathbb{R} \, , \forall f_1,f_2 \in \mathcal{D}.
    \end{equation*}
    
\end{enumerate}
\end{definition}

Analogously to the ground state case, the KMS one-particle structure is also unique up to unitary equivalence~\cite[Thm.~1b]{Kay:1985yx}.\\

Assume now that a ground one-particle structure exists and note that the regularity condition $K^\infty\mathcal{D}\subset \operatorname{dom}(h^{-1/2})$ is trivially satisfied since $h^{-1}$ is bounded by assumption. Further assume that it exists a preferred complex conjugation $\complexconj$ on $\mathscr{H}^\infty$ such that $\complexconj e^{-ith}=e^{ith}\complexconj$. Then, the (unique) KMS one-particle structure is obtained by doubling the ground one-particle structure. More in details, we have $\mathscr{H}^\beta=\mathscr{H}^\infty\oplus\mathscr{H}^\infty$ and the action of $K^\beta$ is 
\begin{equation}\label{eq: formaKbeta}
     K^{\beta} f = \complexconj \sinh (Z^{\beta}) f \oplus \cosh (Z^{\beta}) f,\quad f \in\mathcal{D},
\end{equation}
where $\implicitop$ is defined implicitly on $\Hil^\infty$ by $\tanh (\implicitop) = e^{-\beta \opham /2}$. Finally, it holds 
\begin{equation*}
        e^{-it\Tilde{h}} = e^{it h}\oplus e^{-it h}.
\end{equation*}
Accordingly, the complex Hilbert space of the GNS representation of $\mathcal{A}(\mathcal{D},\sigma)$ induced by the quasi-free state $\omega^\beta$ is
\begin{equation}\label{Eq: KMS_fock}
    \mathscr{H}_{\omega^{\beta}} = \bigoplus_{n=0}^{\infty} (\mathscr{H}^\infty\oplus \mathscr{H}^\infty)^{\otimes_s n},
\end{equation}
and the vacuum vector $\Omega_\beta\in(\mathscr{H}^\infty\oplus\mathscr{H}^\infty)^{\otimes_s0}=\mathbb{C}\oplus\mathbb{C}$ implements the KMS state $\omega^\beta$ on the represented Weyl quasi-local algebra. In addition, the vector state $\Omega_\beta$ extends to a pure state on the irreducible $\operatorname{C}^*$-algebra $\mathcal{B}(\mathscr{H}_{\omega^\beta})$ and, for this reason, the construction we have shortly described is often referred to in the literature as purification. 

\vspace{3mm}

Given a representation $\pi_\omega$ of the quasi-local Weyl algebra $\mathcal{A}(\mathbb{M})$ and the corresponding represented net of $\operatorname{C}^*$-algebras, we construct a net of von Neumann algebras by defining the assignment
\begin{equation*}
    \mathcal{O}\mapsto \mathcal{M}_\omega(\Ocal)\coloneqq\pi_\omega(\mathcal{A}(\Ocal))'',
\end{equation*}
where $\pi_\omega(\mathcal
A(\Ocal))'$ denotes the commutant of $\pi_\omega(\mathcal
A(\Ocal))$, i.e. the set of operators in $\mathcal{B}(\mathscr{H}_\omega)$ that commutes with every element of $\pi_\omega(\mathcal{A}(\Ocal))$.
Clearly this net continues to satisfy the causal structure, namely $\mathcal{M}_\omega(\Ocal_1)\subset \mathcal{M}_\omega(\Ocal)' $ if $\Ocal_1\subset\Ocal'$, where $\Ocal'$ denotes the interior of the causal complement of $\Ocal$. We also denote by $\mathcal{M}_\omega(\mathbb{M})=\pi_\omega(\mathcal{A}(\mathbb{M}))''$. If the state $\omega$ is pure, so that the corresponding GNS representation is irreducible, we have $\pi_\omega(\mathcal{A})''=\mathcal{B}(\mathscr{H}_\omega)$. This is the case of the ground state representation of the scalar field, for which $\mathcal{M}_{\omega^\infty}(\mathbb{M})\equiv\mathcal{M}_\infty(\mathbb{M})=\mathcal{B}(\mathscr{H}_{\omega^\infty})$. In this particular model, for spacetime regions $\Ocal \subset \mathbb{M}$ corresponding to causal diamonds, Haag duality holds~\cite{ArakiHD2}. Namely
\begin{equation*}
    \mathcal{M}_{\omega^\infty}(\Ocal)'\equiv\mathcal{M}_\infty(\Ocal)'=\mathcal{M}_\infty(\Ocal'),
\end{equation*}
where $\mathcal{M}_\infty(\Ocal')$ denotes the double commutant of the quasi-local algebra induced by the $\operatorname{C}^*$-inductive limit of the represented $\operatorname{C}^*$-algebras localized in $\Ocal'$.

\subsection{Main result}
In the thermal case, instead, the situation is quite different. Since the thermal equilibrium state is not pure, the corresponding GNS representation is reducible. The vector $\Omega_{\omega^\beta}$ extends the state $\omega^\beta$ to a normal state on the von Neumann algebra $\mathcal{M}_{\omega^\beta}(\mathbb{M})$. If the time evolution is implemented in the representation by a strongly continuous one-parameter unitary group~\footnote{Note that this condition is matched if the representation coincides with the one constructed over a KMS one-particle structure.}, the vector state $\Omega_{\omega^\beta}$ defines a KMS state on $\mathcal{M}_{\omega^\beta}(\mathbb{M})\equiv \mathcal{M}_\beta(\mathbb{M})$ (see the proof of~\cite[Cor.~5.3.4.]{BraRob:qsm2}) and, therefore, it is not only cyclic (by GNS construction), but also separating for the represented quasi-local von Neumann algebra~\cite[Cor.~5.3.9.]{BraRob:qsm2}. Consequently Tomita-Takesaki modular theory for the pair $(\mathcal{M}_{\beta}(\mathbb{M}),\Omega_{\omega^\beta})$ can be constructed and we denote by $J,\Delta$ the corresponding modular data. The commutant $\mathcal{M}_{\beta}(\mathbb{M})'$ is then obtained as
\begin{equation*}
    \mathcal{M}_{\beta}(\mathbb{M})'=J\mathcal{M}_{\beta}(\mathbb{M})J.
\end{equation*}
Consequently, the following relation holds
\begin{equation*}
    J\mathcal{M}_{\beta}(\mathbb{M})J\vee \mathcal{M_\beta}(\Ocal')=\left(J\mathcal{M}_{\beta}(\mathbb{M})J\cup\mathcal{M_\beta}(\Ocal')\right)''\subset \mathcal{M}_\beta(\Ocal)'.
\end{equation*}
Here $\vee$ denotes the von Neumann algebra generated by the set-theoretic union, $\mathcal{M}_\beta(\Ocal)\equiv\mathcal{M}_{\omega^\beta}(\Ocal)$ and $\mathcal{M}_{\beta}(\Ocal')$ is defined analogously to the ground state case. The main goal of this paper is to identify sufficient conditions that imply the following generalised form of Haag duality
\begin{equation*}
    J\mathcal{M}_{\beta}(\mathbb{M})J\vee \mathcal{M_\beta}(\Ocal')=\mathcal{M}_\beta(\Ocal)',
\end{equation*}
assuming that Haag duality holds in the ground representation. Our main theorem is the following

\begin{theorem}\label{thm: main}
Let $\Ocal \mapsto \mathcal{A}(\Ocal)$ be the abstract net of Weyl $\operatorname{C}^*$-algebras associated with a free real massive scalar field. Let $\Cauchy_t$ be the Cauchy surface at fixed time coordinate $t$ and $\Cdiamond(\subMink)$ be a (open) causal diamond with base $\subMink\subset\Cauchy_t$\footnote{For the precise definition see Equation \eqref{eq: C(B)}}. Consider $\omega^{\beta}$, for $0< \beta < \infty$, the KMS state with respect to the free dynamics on the quasi-local algebra $\mathcal{A}(\mathbb{M})$ and the correspondingly induced net of von Neumann algebras $\mathcal{O} \mapsto \mathcal{M}_\beta(\Ocal)$
\begin{equation*}
    \mathcal{M}_\beta(\Ocal) \coloneqq \pi_{\omega^{\beta}}(\mathcal{A}(\Ocal))''.
\end{equation*}
Then, generalised Haag duality holds in the following sense:
\begin{equation*}
    \mathcal{M}_\beta(\Cdiamond(\subMink))' = \mathcal{M_\beta}(\Cdiamond(\subMink)') \vee J\mathcal{M}_{\beta}(\mathbb{M})J.
\end{equation*}
\end{theorem}

\section{General results for the thermal one-particle structure}\label{sec: General}

Although the main goal of this paper is to study the duality of the local net of von Neumann algebras in the KMS representation associated with a scalar quantum field theory on a globally hyperbolic spacetime manifold, in this section we investigate its general properties without referring to the specific features of the underlying symplectic space. Our results are model-independent in the sense that we analyse only the relative positions of the corresponding one-particle Hilbert space structures and examine how they are affected by a general Bogoliubov transformation.

\subsection{Weyl and Segal formulation of the CCR}\label{sec: RealfromComp}
Since it will be convenient to work within the Weyl formulation of the representation of the CCR algebra, we shortly recall the standard construction of a real Hilbert space $\RealHil$ starting from a complex Hilbert space $\Hil$ with complex inner product $\ComplexProduct{\cdot}{\cdot}_\Hil$ (antilinear in the first entry). 
The real vector space $\RealHil$ coincides with $\Hil$ as a set, but it is given the structure of a vector space over the real field $\mathbb{R}$. The real inner product is $\RealProduct{\cdot}{\cdot}_\RealHil=\real\ComplexProduct{\cdot}{\cdot}_\Hil$ and, since obviously $\|\cdot\|_\RealHil=\|\cdot\|_\Hil$, $\RealHil$ is complete in the norm topology so that it defines a real Hilbert space. The complex part of the inner product defines a symplectic form $\Sympl{\cdot}{\cdot}_\RealHil=2\immag\ComplexProduct{\cdot}{\cdot}_\Hil$ on $\RealHil$. The real and complex part of the inner product are related by a unique (by the Riesz representation theorem) canonical complex structure $\ComplexStr_\RealHil:\RealHil\mapsto\RealHil$, $\ComplexStr_\RealHil=-\ComplexStr_\RealHil^*$ (where $^*$ denotes here the adjoint w.r.t. $\RealProduct{\cdot}{\cdot}_\RealHil$), $\ComplexStr_\RealHil^2=-\mathbbm{1}_\RealHil$ by
\begin{equation}\label{Eq: cs_relates_scalarproducts}
    \frac{1}{2}\Sympl{v}{w}_\RealHil=\immag\ComplexProduct{v}{w}_\Hil=-\real \ComplexProduct{v}{
    iw}_\Hil=\RealProduct{v}{\ComplexStr_\RealHil w}_\RealHil,\quad\forall v,w\in\RealHil.
\end{equation}
We finally denote by $\RealSub{}$ a closed real linear subspace of $\RealHil$ such that $\RealSub{}\perp\ComplexStr_\RealHil\RealSub{}$ (orthogonal with respect to $\RealProduct{\cdot}{\cdot}_\RealHil$) and $\RealHil=\RealSub{}\oplus\beta_{\RealHil{}}\RealSub{}$ (this decomposition always exists, but is not unique, see~\cite[Footnote~13]{ArakiHD1})\footnote{For $A$ and $B$ linear subspaces of a Hilbert space $C$ we use the symbol $A+ B$ to denote the (not necessarily closed) subspace generated by the linear span of $A\cup B$. If in addition $A\cap B=\{ 0 \}$, we use the notation $A\oplus B$. Note that this notation does not necessarily imply that $A\perp B$; if $A\perp B$, it holds that $\overline{A\oplus B}=\overline{A}\oplus\overline{B}$. Finally, we also use the notation $C\oplus D$ to denote the external direct sum of Hilbert spaces and correspondingly the notation $A\oplus B$ for the direct sum of $A\oplus\{0\}\subset C\oplus D$ and $\{0\}\oplus B \subset C\oplus D$. When this notation could be source of confusion, we make use of the symbol $\oplus_\mathrm{e}$ to denote the external direct sum.}.\\

This standard construction applies also to a doubled complex Hilbert space $\Hil\oplus\Hil$ with inner product $\ComplexProduct{\cdot}{\cdot}$. We denote by $\RealProduct{\cdot}{\cdot}$ the real scalar product on the corresponding real Hilbert space, which clearly coincides with $\RealHil\oplus\RealHil$. The complex structure $\ComplexStr_{\RealHil\oplus\RealHil}$ on $\RealHil\oplus\RealHil$ is related to $\ComplexStr_\RealHil$ by $\ComplexStr_{\RealHil\oplus\RealHil}=\ComplexStr_\RealHil\oplus\ComplexStr_\RealHil$. Indeed $\ComplexStr_\RealHil\oplus\ComplexStr_\RealHil$ is obviously a linear map from $\RealHil\oplus\RealHil$ to itself, satisfying $(\ComplexStr_\RealHil\oplus\ComplexStr_\RealHil)^2=-\mathbbm{1}$, $(\ComplexStr_\RealHil\oplus\ComplexStr_\RealHil)^* = -\ComplexStr_\RealHil\oplus\ComplexStr_\RealHil$ (with the adjoint defined with respect to the inner product $\RealProduct{\cdot}{\cdot}$) and for every $v_1,v_2,w_1,w_2\in\RealHil$ it holds
\begin{align*}
    \frac{1}{2}\Sympl{v_1\oplus v_2}{w_1\oplus w_2}&\coloneqq\immag\ComplexProduct{v_1\oplus v_2}{w_1\oplus w_2}\\
    &=\immag\ComplexProduct{v_1}{w_1}_\Hil+\immag\ComplexProduct{v_2}{w_2}_\Hil\\
    &=\RealProduct{v_1}{\ComplexStr_\RealHil w_1}_\RealHil+\RealProduct{v_2}{\ComplexStr_\RealHil w_2}_\RealHil\\
    &=\RealProduct{v_1\oplus v_2}{\ComplexStr_\RealHil\oplus \ComplexStr_\RealHil (w_1\oplus w_2)}\\
    &=\RealProduct{v_1\oplus v_2}{\ComplexStr_{\RealHil\oplus\RealHil} (w_1\oplus w_2)}.
\end{align*}
The decomposition $\RealHil\oplus\RealHil=(\RealSub{}\oplus_\mathrm{e}\RealSub{})\oplus \ComplexStr_{\RealHil\oplus\RealHil}(\RealSub{}\oplus_\mathrm{e}\RealSub{})$ holds, with $(\RealSub{}\oplus_\mathrm{e}\RealSub{})\perp \ComplexStr_{\RealHil\oplus\RealHil}(\RealSub{}\oplus_\mathrm{e}\RealSub{})$ w.r.t. the real scalar product $\RealProduct{\cdot}{\cdot}$.\\

Following~\cite[Sec.~2]{ArakiHD1} we now make use of the spaces $\Hil$, $\RealHil$ and $\RealSub{}$ to describe two equivalent formulations of the Fock space representation of the canonical commutation relations (CCR). First let us recall that the symmetrised Fock space $\Fock{\Hil}$ is defined as
\begin{equation*}
    \Fock{\Hil}=\bigoplus_{n=0}^\infty\Hil^{\otimes_sn},
\end{equation*}
with $\Hil^{\otimes_s0}=\mathbb{C}$. On this Hilbert space we define the usual creation and annihilation operators $a^\dagger(v),a(w),\;v,w\in\Hil$
\begin{align*}
    a^\dagger(v)v_1\otimes_s\ldots \otimes_s v_n&=(n+1)^{1/2}v\otimes_sv_1\otimes_s\ldots \otimes_s v_n;\\
    a(w)v_1\otimes_s\ldots \otimes_s v_n&=n^{-1/2}\sum_{i=1}^n\ComplexProduct{w}{v_i}_{\Hil}v_1\otimes_s\ldots\widehat{v_i}\ldots \otimes_s v_n;\\
    a^\dagger(v)\Omega&=v;\qquad a(w)\Omega=0,
\end{align*}
where $\Omega\in\Hil^{\otimes_s0}$ is the vacuum vector and only in this case $\widehat{h_i}$ means that vector is to be omitted. Note that the annihilation operator is by definition antilinear in its entry $w$. The linear extension of these operators is defined on the dense domain $F_0$ of vectors in $\Fock{\Hil}$ having finitely many non vanishing components. On this domain the operators satisfy the CCR
\begin{equation}\label{Eq:CCR_unbound}
    [a(w),a^\dagger(v)]\Psi=\ComplexProduct{w}{v}_\Hil\Psi,\quad\Psi\in F_0.
\end{equation}
The field operator $\phi(v)$ with dense domain $F_0$ is defined as
\begin{equation*}
    \phi(v)=a(v)+a^{\dagger}(v).
\end{equation*}
The set $F_0$ is a dense set of analytic vectors for $\phi(v)$. Therefore, Nelson's theorem implies that $\phi(v)$ is essentially self-adjoint and we denote with the same symbol its unique self-adjoint extension. 

In the Segal formulation of the Fock representation of the CCR we construct the representation starting from the real Hilbert space $\RealHil$ and the complex structure $\ComplexStr_{\RealHil}$ associated to $\Hil$. We consider the filed operator $\phi(v)$, but we now identify its argument $v$ with an element of $\RealHil$. The commutation relations~\eqref{Eq:CCR_unbound} needs then to be replaced by
\begin{equation*}
    [a(w),a^\dagger(v)]\Psi=\RealProduct{w}{v}_\RealHil\Psi+i\RealProduct{w}{\ComplexStr_\RealHil v}_\RealHil \Psi,\quad\Psi\in F_0.
\end{equation*}
The Weyl operators $W(v)=e^{i\phi(v)}$ are the unitary operators generated by the self-adjoint fields $\phi(v)$ and it follows~\cite[Prop.~5.2.4.]{BraRob:qsm2} that they satisfy the CCR in the Segal form
\begin{equation*}
W(v)W(w)=W(v+w)e^{-i\RealProduct{v}{\ComplexStr_\RealHil w}_\RealHil}=W(v+w)e^{-\frac{i}{2}\Sympl{v}{w}_\RealHil}.
\end{equation*}

As already mention before, it always exists a (non unique) real Hilbert space $\RealSub{}\subset\RealHil$ such that $\RealHil=\RealSub{}\oplus\ComplexStr_\RealHil\RealSub{}$. For any $v_1,v_2\in\RealSub{}$ we define the operators
\begin{align*}
\varphi(v_1)&=a^\dagger(v_1)+a(v_1),\\
\pi(v_2)&=i(a^\dagger(v_2)-a(v_2))
\end{align*}
and, from now on, we use the same symbol to denote their unique self-adjoint closure. On the subspace $F_0$ they satisfy the commutation relations
\begin{align*}
    [\varphi(v_1),\varphi(v_2)]\Psi &=[\pi(v_1),\pi(v_2)]\Psi=0,\quad \Psi\in F_0,\\
    [\varphi(v_1),\pi(v_2)]\Psi&=2i\RealProduct{v_1}{v_2}_\RealHil\Psi,\quad\Psi\in F_0
\end{align*}
and consequently the unitary operators $U(v_1)\coloneqq e^{i\varphi(v_1)}$ and $V(v_2)=e^{i\pi(v_2)}$ satisfy 
\begin{equation*}
U(v_1)V(v_2)U(v_3)V(v_4)=U(v_1+v_3)V(v_2+v_4)e^{2i\RealProduct{v_2}{v_3}_\RealHil},
\end{equation*}
namely the CCR in the Weyl form. Given a generic element $v\in\RealHil$, it admits a unique decomposition $v=v_1+\ComplexStr_\RealHil v_2$, with $v_1,v_2\in\RealSub{}$, and the following equalities relating the field operators in the Segal and Weyl formulation are easily verified
\begin{align}\label{Eq: connection_weyl_segal}
\phi(v)&=\varphi(v_1)+\pi(v_2),\\
    W(v)&=U(v_1)V(v_2)e^{i\RealProduct{v_1}{v_2}_\RealHil}.
\end{align}

Let now $\RealSub{1},\RealSub{2}$ be two linear subspaces of $\RealSub{}$ and $\mathcal{H}_1$ a linear subspace of $\mathcal{H}$. Denoting by
\begin{equation}\label{eq: SegalWeyl}
    \mathcal{R}_{\mathrm{S}}(\mathcal{H}_1)\coloneqq \left\{W(v);v\in\mathcal{H}_1\right\}'';\quad \mathcal{R}_\mathrm{F}(\RealSub{1},\RealSub{2})\coloneqq \left\{U(v_1)V(v_2); v_1\in\RealSub{1},v_2\in\RealSub{2}\right\}''
\end{equation}
the Segal and Weyl von Neumann subalgebra associated to the subspaces, Equation~\eqref{Eq: connection_weyl_segal} shows that
\begin{equation}\label{eq: chiusuravN}
 \mathcal{R}_\mathrm{F}(\RealSub{1},\RealSub{2})=\mathcal{R}_\mathrm{S}(\RealSub{1}\oplus\ComplexStr_\RealHil\RealSub{2}).
\end{equation}
Vice versa, it can be proven~\cite[Thm.~3.]{ArakiHD1} that for any $\mathcal{H}_1\subset\RealHil$ there exists subspaces $\RealSub{},\RealSub{1},\RealSub{2}$ of $\RealHil$ such that $\RealHil=\RealSub{}\oplus\ComplexStr_\RealHil\RealSub{}$, $\mathcal{H}_1=\RealSub{1}\oplus\ComplexStr_\RealHil\RealSub{2}$ with $\RealSub{1},\RealSub{2}\subset\RealSub{}$. Therefore, we can always decide if to work with the Weyl or with Segal formulation of these subalgebras depending on which of the two formulations is more convenient. 
Note that the following equalities hold
\begin{equation}\label{Eq: tacke_closure}
    \mathcal{R}_\mathrm{S}(\mathcal{H}_1)=\mathcal{R}_\mathrm{S}(\overline{\mathcal{H}_1}), \quad \mathcal{R}_\mathrm{F}(\RealSub{1},\RealSub{2})=\mathcal{R}_\mathrm{F}(\overline{\RealSub{1}},\overline{\RealSub{2}}),
\end{equation}
where the overline denotes the closure in the norm topology of $\RealHil$. In order to prove the previous equalities first observe that the operator $W(w)$ is strongly continuous in $w$ with respect to the norm topology of $\RealHil$, and the operators $U(v_1),V(v_2)$ are strongly continuous in $v_1,v_2$ with respect to the norm topology of $\RealSub{}$. This standard result follows from the fact that the self-adjoint generators $\phi(v),\varphi(v_1),\pi(v_2)$ are strongly continuous in their argument on the core $F_0$ (see for instance~\cite[Prop.~5.2.4.,(4)]{BraRob:qsm2}). Since by the bicommutant theorem a von Neumann algebra is closed in the strong topology, this observation shows that the generators of the von Neumann algebras on the left and on the right hand side of Equations~\eqref{Eq: tacke_closure} coincide. 

\subsection{Results on one-particle Hilbert spaces and duality}\label{sec: dualityabstract}
Let $\Hil$ be a complex Hilbert space, $h$ a strictly positive (possibly unbounded) operator on $\Hil$ and $\complexconj$ a complex conjugation on $\Hil$; assume that $\complexconj$ and $h$ commute, in the sense that $\complexconj e^{ith}=e^{-ith}\complexconj$. As mentioned in Section~\ref{Sec: RealScalarField}, if $\Hil$ describes a ground one-particle Hilbert space, then under the previous assumptions the KMS one-particle Hilbert space coincides with $\Hil\oplus\Hil$. 
Since $h$ is strictly positive, the operator $e^{-\beta h},\;\beta\in [0,+\infty)$ is bounded and self-adjoint.

As explained in the previous section, we denote by $\RealHil$ the restriction of $\Hil$ to the field $\mathbb{R}$ and by $\ComplexStr_\RealHil$ the induced complex structure on $\RealHil$. Let $\RealSub{}\in\RealHil$ be a (non-unique) closed subspace such that $\RealHil=\RealSub{}\oplus\ComplexStr_\RealHil\RealSub{}$. We assume that $\RealSub{}$ is \textit{invariant} under the action of $e^{-\beta h}$ and $\complexconj$, identified as bounded self-adjoint operators on $\RealHil$, i.e.
\begin{equation}
    e^{-\beta h}P_{\RealSub{}}= P_{\RealSub{}}e^{-\beta h}P_{\RealSub{}},\qquad \complexconj P_{\RealSub{}}=P_{\RealSub{}}\complexconj P_{\RealSub{}},
\end{equation}
where $P_{\RealSub{}}$ denotes the unique orthogonal projector on $\RealSub{}$. This implies that $[e^{-\beta h},P_{\RealSub{}}]=0$ and, by the spectral theorem, $[\sinh (\implicitop),P_{\RealSub{}}]=[\cosh (\implicitop),P_{\RealSub{}}]=0$ with
\begin{equation*}
    \cosh(\implicitop)=\frac{1}{\sqrt{1-e^{-\beta h}}},\qquad\quad \sinh(\implicitop)= \frac{1}{\sqrt{e^{\beta h}-1}}=\frac{e^{-\frac{\beta h}{2}}}{\sqrt{1-e^{-\beta h}}}\;.
\end{equation*}
In particular, the closed subspace $\RealSub{}$ is invariant under the action of the operators $\sinh(\implicitop),\cosh (\implicitop)$. 

In this setting, we consider two real subspaces $\RealSub{1},\RealSub{2}\subset\RealSub{}$ and their image in the doubled real Hilbert space $\RealHil\oplus\RealHil$ under the action of a bounded Bogoljubov transform. First of all, we prove the following proposition concerning their orthogonal in $\RealSub{}\oplus\RealSub{}$   
\begin{proposition}\label{prop: GenericSub}
    Let $\RealSub{i}\subset\RealSub{}$, for $i= 1,2$, be a real subspace and $\RealSub{i}^\perp\in\RealSub{}$ be its orthogonal in $\RealSub{}$. Let
    \begin{align*}
         \mathcal{U}(\RealSub{1})&\coloneqq\left\{u_1\oplus u_2\in\RealSub{}\oplus\RealSub{}:u_1=\complexconj\sinh (\implicitop)u, u_2 = \cosh (\implicitop) u,\; u\in\RealSub{1}\right\}.\\
         \mathcal{V}(\RealSub{2})&\coloneqq\left\{u_1\oplus u_2\in\RealSub{}\oplus\RealSub{}:u_1=-\complexconj\sinh (\implicitop)u, u_2 = \cosh (\implicitop) u,\; u\in\RealSub{2}\right\}.\\
         \tilde{\mathcal{V}}&\coloneqq\left\{v_1\oplus v_2\in\RealSub{}\oplus\RealSub{}:v_1=\cosh (\implicitop) v, v_2 = -\complexconj\sinh (\implicitop)v,\; v\in\RealSub{}\right\}.\\
         \tilde{\mathcal{U}}&\coloneqq\left\{v_1\oplus v_2\in\RealSub{}\oplus\RealSub{}:v_1=\cosh (\implicitop) v, v_2 = \complexconj\sinh (\implicitop)v,\; v\in\RealSub{}\right\}.
    \end{align*}
    Then 
    \begin{align*}
      \mathcal{U}(\RealSub{1})^\perp=&\;\mathcal{V}(\RealSub{1}^\perp)\oplus\tilde{\mathcal{V}},\\
      \mathcal{V}(\RealSub{2})^\perp=&\;\mathcal{U}(\RealSub{2}^\perp)\oplus \tilde{\mathcal{U}},
    \end{align*}
    where $\mathcal{U}(\RealSub{1})^\perp,\mathcal{V}(\RealSub{2})^\perp$ denote the orthogonal subspaces in $\RealSub{}\oplus\RealSub{}$.
\end{proposition}
\begin{proof}
  Let us define the real-linear operator $A$ 
    \begin{equation*}
        A \coloneqq\begin{pmatrix}
            \cfrac{1}{\sqrt{1-e^{-\beta h}}} &-\complexconj\cfrac{e^{-\frac{\beta h}{2}}}{\sqrt{1-e^{-\beta h}}} \\
            -\complexconj\cfrac{e^{-\frac{\beta h}{2}}}{\sqrt{1-e^{-\beta h}}} & \cfrac{1}{\sqrt{1-e^{-\beta h}}} 
        \end{pmatrix},
    \end{equation*}
    on the Hilbert space $\Hil\oplus\Hil$. This operator is bounded, with
        \begin{equation*}
            \operatornorm{A}\leq 2\sup_{i,j\in[1,2]}\operatornorm{A_{ij}}=2\frac{1}{\sqrt{1-e^{-\beta m}}}, \quad m=\min\spec{h},
        \end{equation*}
        with $\spec{h}$ denoting the spectrum of the operator $h$. The operator $A$ is invertible with bounded inverse
        \begin{equation*}
            A^{-1} \coloneqq\begin{pmatrix}
            \cfrac{1}{\sqrt{1-e^{-\beta h}}} &\complexconj\cfrac{e^{-\frac{\beta h}{2}}}{\sqrt{1-e^{-\beta h}}} \\
            \complexconj\cfrac{e^{-\frac{\beta h}{2}}}{\sqrt{1-e^{-\beta h}}} & \cfrac{1}{\sqrt{1-e^{-\beta h}}},
        \end{pmatrix},\quad \operatornorm{A^{-1}}\leq\frac{1}{\sqrt{1-e^{-\beta m}}}.
        \end{equation*}
        From now on we identify $A,A^{-1}$ with the real-linear operators that they induce on $\RealHil\oplus\RealHil$. On $\RealHil\oplus\RealHil$ the operators $A,A^{-1}$ are self-adjoint. In addition,
        since $\RealSub{}$ is an invariant subspace of $e^{-\beta h}$ and $\complexconj$ (also identified with the real-linear operators induced on $\RealHil$), it follows by the spectral theorem that $\RealSub{}\oplus\RealSub{}$ is an invariant subspace of $A,A^{-1}$. We denote by $A_{\RealSub{}},A^{-1}_{\RealSub{}}$ the restriction of $A,A^{-1}$ to the Hilbert space $\RealSub{}\oplus \RealSub{}$. Recalling that $\complexconj e^{i\beta h}=e^{-i\beta h}\complexconj$
        \begin{equation*}
            A_{\RealSub{}}\mathcal{U}(\RealSub{1})=\left\{0\right\}\oplus_\mathrm{e}\RealSub{1}.
        \end{equation*}
        Its orthogonal in $\RealSub{}\oplus\RealSub{}$ is
        \begin{equation*}
             \left(A_{\RealSub{}}\mathcal{U}(\RealSub{1})\right)^\perp=\RealSub{}\oplus_\mathrm{e}\RealSub{1}^\perp.
        \end{equation*}
        We now apply Lemma~\ref{Lem: bounded_inverse} to the operator $A_{\RealSub{}}$ and conclude that
        \begin{equation*}
            \mathcal{U}(\RealSub{1})^\perp=A_{\RealSub{}}\left(A_{\RealSub{}}\mathcal{U}(\RealSub{1})\right)^\perp=\mathcal{V}(\RealSub{1}^\perp)+\tilde{\mathcal{V}}=\mathcal{V}(\RealSub{1}^\perp)\oplus\tilde{\mathcal{V}},
        \end{equation*}
       where the last equality follows since $A$ is injective. 
        The second equality is proven in the same way, substituting $A_{\RealSub{}}$ with $A^{-1}_{\RealSub{}}$.
\end{proof}

The proof of the duality property of a net of Weyl algebras is simplified if the closed real subspaces $\RealSub{1},\RealSub{2}$ that label the algebra are in \textit{generic position}~\cite{EckmannOsterwalder1973}. By this we mean that the intersection of any pair of the real closed subspaces (of $\RealSub{}$) $\RealSub{1},\RealSub{2},\RealSub{1}^\perp,\RealSub{2}^\perp$ are trivial, i.e.~equal to $\{0\}$. This condition is generally matched when considering local von Neumann algebras associated to a bounded open region $\Ocal$ of Minkowski spacetime in the ground representation of a free massive real scalar field. However, in the present more generic situation the relative position between the corresponding real subspaces is not necessarily preserved. More in details we have
\begin{proposition}\label{prop: nongeneric}
    Let $\RealSub{1},\RealSub{2}\subset\RealSub{}$ be closed real subspaces of $\RealSub{}$. Then $\mathcal{U}(\RealSub{1}),\mathcal{V}(\RealSub{2})$ are closed real subspaces of $\RealSub{}\oplus\RealSub{}$ and
    \begin{itemize} 
    \item $\mathcal{U}(\RealSub{1})\wedge\mathcal{V}(\RealSub{2})=\{0\} \oplus \{0\}$ (even if $\RealSub{1}\wedge\RealSub{2}\not\subset\{0\}$). 
    \item $\mathcal{U}(\RealSub{1})\wedge\mathcal{V}(\RealSub{2})^\perp=\{0\}\oplus\{0\}$ iff $\RealSub{1}\wedge \RealSub{2}^\perp=\{0\}$ and $\mathcal{U}(\RealSub{1})^\perp\wedge\mathcal{V}(\RealSub{2})=\{0\}\oplus\{0\}$ iff $\RealSub{1}^\perp\wedge \RealSub{2}=\{0\}$.
    \item $\RealSub{1},\RealSub{2}$ in generic position $\centernot\implies$ $\mathcal{U}(\RealSub{1})^\perp\wedge \mathcal{V}(\RealSub{2})^\perp =\{0\}\oplus\{0\}$.
    \end{itemize}
\end{proposition}

\begin{proof}
    The subspaces $\mathcal{U}(\RealSub{1}),\mathcal{V}(\RealSub{2})$ coincide with $A^{-1}_{\RealSub{}}(\{0\}\oplus\RealSub{1}),A_{\RealSub{}}(\{0\}\oplus\RealSub{2}))$ and are closed since both $A_{\RealSub{}}$ and $A_{\RealSub{}}^{-1}$ are bounded operators.\\
   We prove the first statement. The operators $\cosh (\implicitop),\sinh (\implicitop)$ are obviously injective, so that
    \begin{align*}
        \cosh( {\implicitop})u&=\cosh (\implicitop)v,\quad u,v\in\RealSub{}\implies u=v;\\
        \complexconj\sinh( {\implicitop})u&=-\complexconj\sinh (\implicitop)v,\quad u,v\in\RealSub{}\implies u=-v.
    \end{align*}
     By definition of $\mathcal{U}(\RealSub{1}),\mathcal{V}(\RealSub{2})$ it follows that they have trivial intersection. \\
     We prove the second statement. Applying the operator $A_{\RealSub{}}$ on $\mathcal{U}(\RealSub{1})\wedge\mathcal{V}(\RealSub{2})^\perp$ we obtain
     \begin{equation*}
         A_{\RealSub{}}\left(\mathcal{U}(\RealSub{1})\wedge\mathcal{V}(\RealSub{2})^\perp\right)=\left(\{0\}\oplus\RealSub{1}\right)\wedge (\RealSub{}\oplus \RealSub{2}^\perp)= \{0\} \oplus (\RealSub{1}\wedge\RealSub{2}^\perp),
     \end{equation*}
     where the first equality follows since $A_{\RealSub{}}$ is injective. Once again, since $A_{\RealSub{}}$ is injective this proves the first part of the statement. The second part is proven in the same way using the operator $A^{-1}_{\RealSub{}}$. \\
     In order to prove the final statement we show that $(\mathcal{U}(\RealSub{1})^\perp\wedge \mathcal{V}(\RealSub{2})^\perp)^\perp$ is not necessarily dense in $\RealSub{}\oplus\RealSub{}$
     \begin{equation*}
         \RealSub{} \oplus \RealSub{}\not\subseteq \overline{\mathcal{U}(\RealSub{1}) + \mathcal{V}(\RealSub{2})}.
     \end{equation*}
     It is enough to show it in a particular case, so let us assume $h$ to be further a bounded operator. By definition
     \begin{equation*}
         \mathcal{U}(\RealSub{1})+ \mathcal{V}(\RealSub{2}) = \left\{u_1\oplus u_2\in\RealSub{}\oplus\RealSub{}:u_1=\complexconj\sinh (\implicitop)(u-v), u_2 = \cosh (\implicitop) (u+v),\; u\in\RealSub{1}, v \in \RealSub{2}\right\}
     \end{equation*}
     Since the subspaces $\RealSub{1},\RealSub{2}$ are in generic position, there are vectors $w\in\RealSub{}$ such that $w\not\in\RealSub{1}$.  
     Consider a vector $\psi\in \RealSub{}\oplus\RealSub{}$ of the form
     \begin{equation*}
         \psi=w_1\oplus w_2\in\RealSub{}\oplus\RealSub{},\;w_1=\complexconj\sinh (\implicitop)w, w_2 = \cosh (\implicitop) w,\; w \not\in \RealSub{1}.
     \end{equation*}
     Suppose that $\psi\in \overline{\mathcal{U}(\RealSub{1})+\mathcal{V}(\RealSub{2})}$. It follows that there exist two sequences $\{u_i\}_{i\in\mathbb{N}},\{v_i\}_{i\in\mathbb{N}}$ such that $u_i\in\RealSub{1},v_i\in\RealSub{2}$ and
     \begin{equation*}
        \complexconj\sinh(\implicitop) ( u_i-v_i)\xrightarrow[i\to\infty]{}\complexconj\sin(\implicitop) w,\quad\cosh(\implicitop)( u_i+v_i)\xrightarrow[i\to\infty]{} \cosh(\implicitop)w.
     \end{equation*}
     Since $h$ is bounded, both $\cosh (\implicitop)$ and $\sinh (\implicitop)$ are bounded invertible operators. Therefore, both sequences of vectors $u_i+v_i$ and $u_i-v_i$ are convergent and so
     \begin{equation*}
         \lim_{i\to \infty} u_i=\frac{1}{2}\lim_{i\to \infty} (u_i+v_i+u_i-v_i)=\frac{1}{2}\lim_{i\to \infty} (u_i+v_i)+\frac{1}{2}\lim_{i\to \infty} (u_i-v_i)=w.
     \end{equation*}
     Since $w\not\in\RealSub{1}$, this contradiction proves that $\psi\not\in \overline{\mathcal{U}(\RealSub{1})+\mathcal{V}(\RealSub{2})}$ and therefore the claim.
\end{proof}
\begin{remark}
Note that in the explicit model considered later in Section~\ref{sec: ExMink}, it is well known that the Reeh-Schlieder property applies~\cite{Jakel:1999ji,StrohmaierVerch_2002}. Therefore, it follows that the so called pre-cyclicity property is also satisfied~\cite[Thm.~1.2]{FiglioliniGuido1991} (see also Theorem~\ref{thm: Reeh} in the Appendix), namely
\begin{equation*}
    \overline{\mathcal{U}(\RealSub{1})+\ComplexStr_{\RealHil\oplus\RealHil}\mathcal{V}(\RealSub{2})+i\left(\mathcal{U}(\RealSub{1})+\ComplexStr_{\RealHil\oplus\RealHil}\mathcal{V}(\RealSub{2})\right)}^{\Hil\oplus\Hil}=\Hil\oplus\Hil,
\end{equation*}
and consequently 
\begin{equation*}
    \overline{\mathcal{U}(\RealSub{1})+\mathcal{V}(\RealSub{2})}=\RealSub{}\oplus\RealSub{}.
\end{equation*}
Here $\mathcal{U}(\RealSub{1}),\mathcal{V}(\RealSub{2})$ are the real subspaces that label the von Neumann algebra in the thermal representation associated with a open bounded region of Minkowski.
For consistency with the proof of Proposition~\ref{prop: nongeneric} we observe that in the examples the operator $h$ is not bounded. 
\end{remark}
Following the discussion in Section $4$ and Section $5$ of~\cite{ArakiHD1}, even if the relevant subspaces are not guaranteed to be in generic positions, for the proof of duality we can always reduce to the case in which this condition is satisfied. In particular, one can prove that the represented von Neumann algebras admit tensor product decomposition, where the only factor with a non-trivial commutant consists of (see~\cite[Lemma $5.2$]{ArakiHD1} and the discussion after it)
\begin{equation*}
    \mathcal{R}_{\mathrm{F}}(\mathcal{U}(\RealSub{1})', \mathcal{V}(\RealSub{2})' \backslash \mathcal{U}\mathcal{V}),
\end{equation*}
i.e.~the Weyl algebra in Weyl form labelled by the subspaces $\mathcal{U}(\RealSub{1})', \mathcal{V}(\RealSub{2})' \subset \mathcal{U}\mathcal{V} \subset \RealSub{} \oplus \RealSub{}$ understood as subspaces relative to $\mathcal{U}\mathcal{V}$. 
The definition of $\mathcal{U}\mathcal{V}$ is given in \cite[Section 5]{ArakiHD1} by
\begin{equation*}
    \mathcal{U}\mathcal{V} \coloneqq \left(\mathcal{U}(\RealSub{1}) \wedge \mathcal{V}(\RealSub{2}) + \mathcal{U}(\RealSub{1})^{\perp} \wedge \mathcal{V}(\RealSub{2}) + \mathcal{U}(\RealSub{1}) \wedge \mathcal{V}(\RealSub{2})^{\perp} + \mathcal{U}(\RealSub{1})^{\perp} \wedge \mathcal{V}(\RealSub{2})^{\perp}\right)^{\perp}
\end{equation*}
and correspondingly the closed subspaces
\begin{align*}
    \mathcal{U}(\RealSub{1})' &\coloneqq \mathcal{U}(\RealSub{1}) \wedge \mathcal{U}\mathcal{V} \\
    \mathcal{V}(\RealSub{2})' &\coloneqq \mathcal{V}(\RealSub{2}) \wedge \mathcal{U}\mathcal{V},
\end{align*}
with orthogonal relative to $\mathcal{U}\mathcal{V}$ given by $(\mathcal{U}(\RealSub{1})')^{\perp} \wedge \mathcal{U}\mathcal{V}$ and $(\mathcal{V}(\RealSub{2})')^{\perp} \wedge \mathcal{U}\mathcal{V}$. Then, it is possible to check that relative to $\mathcal{U}\mathcal{V}$ the closed subspaces $\mathcal{U}(\RealSub{1})',\mathcal{V}(\RealSub{2})', (\mathcal{U}(\RealSub{1})')^{\perp} \wedge \mathcal{U}\mathcal{V}$ and $(\mathcal{V}(\RealSub{2})')^{\perp} \wedge \mathcal{U}\mathcal{V}$ all have trivial intersections (see~\cite[Theorem 2']{ArakiHD1}). 
\begin{remark}\label{rem: refer}
Consider the limiting case in which $\RealSub{1} = \RealSub{}$ and $\RealSub{2} = \{ 0 \}$. This corresponds to the abelian von Neumann algebra generated by all and only the represented time zero fields. According to the above definitions it holds
\begin{equation*}
    \mathcal{U}\mathcal{V}  = \{ 0 \oplus 0 \}
\end{equation*}
and the closed subspaces $\mathcal{U}(\RealSub{1})',\mathcal{V}(\RealSub{2})'$ as well as their orthogonal relative to $\mathcal{U}\mathcal{V}$ are trivial.
\end{remark}
Motivated by the applications of the next section, we now assume $\RealSub{1}$ and $\RealSub{2}$ to be in generic position. In this case, by using the results proved in Proposition \ref{prop: nongeneric}, we obtain
\begin{equation}\label{eq: defsubgen}
\begin{aligned}
    &\mathcal{U}\mathcal{V}  = \overline{\mathcal{U}(\RealSub{1}) + \mathcal{V}(\RealSub{2})}\\
    &\mathcal{U}(\RealSub{1})' = \mathcal{U}(\RealSub{1})\\
    &\mathcal{V}(\RealSub{2})' = \mathcal{V}(\RealSub{2}).
\end{aligned}
\end{equation}
As a consistency check, let us show that relative to $\mathcal{U}\mathcal{V}$ the closed subspaces $\mathcal{U}(\RealSub{1})',\mathcal{V}(\RealSub{2})'$ are indeed in generic position. Actually, we have the slightly stronger result
\begin{proposition}\label{prop: generic}
Suppose that $\RealSub{1}\wedge\RealSub{2}^\perp=\{0\}$ and $\RealSub{1}^\perp\wedge\RealSub{2}=\{0\}$. Then the subspaces $\mathcal{U}(\RealSub{1})', \mathcal{V}(\RealSub{2})', (\mathcal{U}(\RealSub{1})')^{\perp} \wedge \mathcal{U}\mathcal{V}, (\mathcal{V}(\RealSub{2})')^{\perp} \wedge \mathcal{U}\mathcal{V}$ as defined above in \eqref{eq: defsubgen}, have all trivial intersections.
\end{proposition}
\begin{proof}
Following the above definition we have
\begin{align*}
    (\mathcal{U}(\RealSub{1})')^{\perp} \wedge \mathcal{U}\mathcal{V} &\coloneqq \mathcal{U}(\RealSub{1})^{\perp} \wedge (\mathcal{U}(\RealSub{1}) \vee \mathcal{V}(\RealSub{2}))\\
    (\mathcal{V}(\RealSub{2})')^{\perp} \wedge \mathcal{U}\mathcal{V} &\coloneqq \mathcal{V}(\RealSub{2})^{\perp} \wedge (\mathcal{U}(\RealSub{1}) \vee \mathcal{V}(\RealSub{2})).
\end{align*}
By the first item of Proposition~\ref{prop: nongeneric} we know that
\begin{equation*}
    \mathcal{U}(\RealSub{1})' \wedge \mathcal{V}(\RealSub{2})' = \{ 0 \} \oplus \{ 0 \}.
\end{equation*}
Using the second time of Proposition~\ref{prop: nongeneric} and the distributivity of $\wedge$ we also get
\begin{equation*}
    \mathcal{U}(\RealSub{1})' \wedge ((\mathcal{V}(\RealSub{2})')^{\perp} \wedge \mathcal{U}\mathcal{V}) = (\mathcal{U}(\RealSub{1}) \wedge \mathcal{V}(\RealSub{2})^{\perp}) \wedge (\mathcal{U}(\RealSub{1}) \vee \mathcal{V}(\RealSub{2})) = \{0\} \oplus \{0\}
\end{equation*}
and analogously
\begin{equation*}
    \mathcal{V}(\RealSub{2})' \wedge ((\mathcal{U}(\RealSub{1})')^{\perp} \wedge \mathcal{U}\mathcal{V}) = \{0\} \oplus \{0\}.
\end{equation*}
Finally, we check that
\begin{align*}
    \left( (\mathcal{U}(\RealSub{1})')^{\perp} \wedge \mathcal{U}\mathcal{V} \right) \wedge \left( (\mathcal{V}(\RealSub{2})')^{\perp} \wedge \mathcal{U}\mathcal{V} \right) &= \left(\mathcal{U}(\RealSub{1})^{\perp} \wedge \mathcal{V}(\RealSub{2})^{\perp}\right) \wedge \left(\mathcal{U}(\RealSub{1}) \vee \mathcal{V}(\RealSub{2})\right)\\
    &= \left(\mathcal{U}(\RealSub{1}) \vee \mathcal{V}(\RealSub{2})\right)^{\perp} \wedge \left(\mathcal{U}(\RealSub{1}) \vee \mathcal{V}(\RealSub{2})\right)\\
    &= \{0\} \oplus \{0\}.
\end{align*}
\end{proof}
Duality is now proven by first applying~\cite[Theorem $2$]{EckmannOsterwalder1973} (reported here as Theorem~\ref{app: EO})
\begin{equation}\label{Eq: gen_pos_comm}
    \mathcal{R}_{\mathrm{F}}\left(\mathcal{U}(\RealSub{1})', \mathcal{V}(\RealSub{2})' \backslash \mathcal{U}\mathcal{V}\right)' = \mathcal{R}_{\mathrm{F}}\left((\mathcal{V}(\RealSub{2})')^{\perp} \wedge \mathcal{U}\mathcal{V}, (\mathcal{U}(\RealSub{1})')^{\perp} \wedge \mathcal{U}\mathcal{V} \backslash \mathcal{U}\mathcal{V}\right),
\end{equation}
valid for subspaces in generic position and then using the results of~\cite[Section $5$]{ArakiHD1}.
\begin{theorem}\label{Thm: GS}
 The following equality holds
 \begin{equation}\label{eq: dualitaF}
    \mathcal{R}_{\mathrm{F}}\left(\mathcal{U}(\RealSub{1}), \mathcal{V}(\RealSub{2})\right)' =\mathcal{R}_\mathrm{F}(\mathcal{U}(\RealSub{2}^\perp),\mathcal{V}(\RealSub{1}^\perp))\vee\mathcal{R}_\mathrm{F}(\tilde{\mathcal{U}},\tilde{\mathcal{V}}).
\end{equation}
\end{theorem}
\begin{proof}
By~\cite[Section $5$]{ArakiHD1}, Equation~\eqref{Eq: gen_pos_comm} implies that
\begin{equation*}
    \mathcal{R}_{\mathrm{F}}\left(\mathcal{U}(\RealSub{1}), \mathcal{V}(\RealSub{2})\right)' = \mathcal{R}_{\mathrm{F}}\left(\mathcal{V}(\RealSub{2})^{\perp}, \mathcal{U}(\RealSub{1})^{\perp}\right).
\end{equation*}
which is equivalent to~\cite[Equation $(3.19)$]{ArakiHD1}.
Using Proposition~\ref{prop: GenericSub}, we obtain
\begin{equation*}
    \mathcal{R}_{\mathrm{F}}\left(\mathcal{U}(\RealSub{1}), \mathcal{V}(\RealSub{2})\right)' = \mathcal{R}_{\mathrm{F}}(\mathcal{U}(\RealSub{2}^\perp)\oplus\tilde{\mathcal{U}},\mathcal{V}(\RealSub{1}^\perp)\oplus\tilde{\mathcal{V}}).
\end{equation*}
Finally, the real linearity of the fields $\varphi,\pi$ (see also~\cite[Theorem 1']{ArakiHD1}) implies the statement of the Theorem.
\end{proof}

Finally, we identify the algebra $\mathcal{R}(\tilde{\mathcal{U}},\tilde{\mathcal{V}})$ with the commutant $\mathcal{R}(\mathcal{U}(\RealSub{}),\mathcal{V}(\RealSub{}))'$
\begin{proposition}\label{prop: modcomm}
In the hypothesis of the previous theorem the following equality holds
 \begin{equation}\label{eq: dualitaF1}
    \mathcal{R}_{\mathrm{F}}\left(\mathcal{U}(\RealSub{1}), \mathcal{V}(\RealSub{2})\right)' =\mathcal{R}_\mathrm{F}(\mathcal{U}(\RealSub{2}^\perp),\mathcal{V}(\RealSub{1}^\perp))\vee \mathcal{R}_\mathrm{F}(\mathcal{U}(\RealSub{}),\mathcal{V}(\RealSub{}))'.
\end{equation}
\end{proposition}
\begin{proof}
By Equation~\eqref{eq: chiusuravN} we consider the von Neumann algebra $\mathcal{R}_\mathrm{S}(\mathcal{U}(\RealSub{})\oplus\ComplexStr_{\RealHil\oplus\RealHil}\mathcal{V}(\RealSub{}))=\mathcal{R}_\mathrm{F}(\mathcal{U}(\RealSub{}),\mathcal{V}(\RealSub{}))$. The real linear closed subspace $\mathcal{U}(\RealSub{})\oplus\ComplexStr_{\RealHil\oplus\RealHil}\mathcal{V}(\RealSub{})\subset\Hil\oplus\Hil$ is a standard subspace, namely
\begin{align*}
       \overline{\mathcal{U}(\RealSub{})\oplus\ComplexStr_{\RealHil\oplus\RealHil}\mathcal{V}(\RealSub{})+i(\mathcal{U}(\RealSub{})\oplus\ComplexStr_{\RealHil\oplus\RealHil}\mathcal{V}(\RealSub{}))}^{\Hil\oplus\Hil}&=\Hil\oplus\Hil,\\ 
       \mathcal{U}(\RealSub{})\oplus\ComplexStr_{\RealHil\oplus\RealHil}\mathcal{V}(\RealSub{})\cap i(\mathcal{U}(\RealSub{})\oplus\ComplexStr_{\RealHil\oplus\RealHil}\mathcal{V}(\RealSub{}))&=\{0\}\oplus\{0\},
\end{align*}
as proven in~\cite{Kay:1985yx,Kay:purification} (see in particular~\cite[App.~2]{Kay:1985yx} and~\cite[Thm.~2]{Kay:purification}, recalling that $\RealSub{}\oplus\ComplexStr_\RealHil \RealSub{}=\RealHil$). We define the closed, antilinear operator $\mathfrak{s}$ over $\Hil\oplus\Hil$ by the following action on its core
   \begin{align*}
       \mathfrak{s}:\; \mathcal{U}(\RealSub{})\oplus\ComplexStr_{\RealHil\oplus\RealHil}\mathcal{V}(\RealSub{})+i(\mathcal{U}(\RealSub{})\oplus\ComplexStr_{\RealHil\oplus\RealHil}\mathcal{V}(\RealSub{}))&\to \mathcal{U}(\RealSub{})\oplus\ComplexStr_{\RealHil\oplus\RealHil}\mathcal{V}(\RealSub{})+i(\mathcal{U}(\RealSub{})\oplus\ComplexStr_{\RealHil\oplus\RealHil}\mathcal{V}(\RealSub{})),\\
       h+ik&\mapsto h-ik.
   \end{align*}
   By the uniqueness of its polar decomposition $\mathfrak{s} = j\delta^{1/2}$, it follows that\footnote{Note that different convention in the definition of $\mathfrak{s}$ and consequently in the overall sign in the operator $j$, if compared with~\cite{Kay:purification}, is consistent with the different convention we adopted in the definition of the field operators compared with~\cite[Sec.~2.3.]{Kay:1985yx}.}
   \begin{equation*}
       \delta^{1/2}=\begin{pmatrix}
           e^{\beta\frac{h}{2}}&0\\
           0&e^{-\beta\frac{h}{2}}
       \end{pmatrix},\quad
       j=\begin{pmatrix}
           0&\complexconj\\
           \complexconj&0
       \end{pmatrix}.     
   \end{equation*}
Since the modular data $J,\Delta$ associated with the Weyl algebra $\mathcal{R}_\mathrm{S}(\mathcal{U}(\RealSub{})\oplus\ComplexStr_{\RealHil\oplus\RealHil}\mathcal{V}(\RealSub{}))$ are the second quantisation of $j,\delta$ (see~\cite{EckmannOsterwalder1973}), Tomita-Takesaki theorem implies that
\begin{equation}
    \mathcal{R}_\mathrm{S}(\mathcal{U}(\RealSub{})\oplus\ComplexStr_{\RealHil\oplus\RealHil}\mathcal{V}(\RealSub{}))'=J\mathcal{R}_\mathrm{S}(\mathcal{U}(\RealSub{})\oplus\ComplexStr_{\RealHil\oplus\RealHil}\mathcal{V}(\RealSub{}))J=\mathcal{R}_\mathrm{S}(j(\mathcal{U}(\RealSub{})\oplus\ComplexStr_{\RealHil\oplus\RealHil}\mathcal{V}(\RealSub{}))).
\end{equation}
To conclude, we now only have to note that
\begin{equation*}
    j(\mathcal{U}(\RealSub{})\oplus\ComplexStr_{\RealHil\oplus\RealHil}\mathcal{V}(\RealSub{}))=\tilde{\mathcal{U}}\oplus\beta_{\RealHil\oplus\RealHil}\tilde{\mathcal{V}},
\end{equation*}
since $\complexconj (\RealSub{}\oplus\ComplexStr_\RealHil\RealSub{})=\RealSub{}\oplus\ComplexStr_\RealHil\RealSub{}$.
\end{proof}
\begin{remark}
For the specific case discussed in Remark~\ref{rem: refer}, Theorem~\ref{Thm: GS} simply proves the well known fact that the commutant of the abelian von Neumann algebra generated by time zero fields is generated by the algebra itself and the commutant of the quasi-local Weyl algebra.
\end{remark}

\section{Scalar field on Minkowski spacetime in the KMS representation}\label{sec: ExMink}
In this section we prove that Haag duality holds for the von Neumann algebras of a real massive scalar field on Minkowski spacetime in a KMS representation. As open bounded regions of spacetime we always consider causal diamonds with base in $\Cauchy_t$. 

\subsection{Fock-KMS representation of local net}
Working on Minkowski spacetime we make use of Fourier integrals to get explicit expressions for the subspaces introduced in the previous sections.\\

As already mentioned, the GNS representation induced by a quasi-free state coincides with the Fock representation constructed on the corresponding one-particle Hilbert space. On Minkowski spacetime, the one-particle ground space is isomorphic to the space of square-integrable functions on the positive mass-hyperboloid $\massh=\left\{p\in\mathbb{R}^4\vert p^\nu p_\nu=-m^2, p^0> 0\right\}$ with respect to the Lorentz invariant measure $\mu_\mathrm{L}$
\begin{equation*}
    \Hil^\infty = L^2(\massh,\di\mu_\mathrm{L}),\quad \di\mu_\mathrm{L}(\textbf{p})=\frac{\di^3\textbf{p}}{(2\pi)^32\en{p}},
\end{equation*}
where $\en{p}=\sqrt{\|\textbf{p}\|^2+m^2}$ and we fixed the mostly plus convention on the signature of the metric. Then, the map $K^{\infty}: C^{\infty}_c(\mathbb{M},\mathbb{R})/\ker (\operatorname{E}) \to \Hil^{\infty}$, introduced in Definition~\ref{def: groundrep} and whose image is a dense subspace in $\Hil^{\infty}$, acts explicitly as follow
\begin{equation*}\label{Eq: fourier_transform_on_mh}
    K^\infty(f)=\left.\hat{f}(p^0,\textbf{p})\right\vert_{\massh}.
\end{equation*}
Following the discussion of Section~\ref{sec: RealfromComp}, we turn the complex Hilbert space $\Hil^{\infty} = L^2(\massh,\di \mu_\mathrm{L})$ into a real Hilbert space $\RealHil^\infty=L^2(\massh,\di\mu_\mathrm{L};\mathbb{R})\oplus L^2(\massh,\di\mu_\mathrm{L};\mathbb{R})$ by identifying every complex valued function $f\in L^2(\massh,\di\mu_\mathrm{L})$ with its real and imaginary parts with respect to the canonical conjugation operator. On $\RealHil^\infty$, the complex structure $\ComplexStr_{\RealHil^\infty}$ that implements the multiplication by $i$ and satisfies the relation~\eqref{Eq: cs_relates_scalarproducts} coincides with the operator
\begin{equation*}
    \ComplexStr_{\RealHil^\infty}=\begin{pmatrix}
        0&-1\\1 &0
    \end{pmatrix}.
\end{equation*}
 $\RealHil^\infty$ can be decomposed as $\RealHil^\infty=\RealSub{}^\infty\oplus \ComplexStr_{\RealHil^\infty}\RealSub{}^\infty$ where $\RealSub{}^\infty=L_\mathrm{s}^2(\massh,\di\mu_\mathrm{L};\mathbb{R})\oplus L_\mathrm{a}^2(\massh,\di\mu_\mathrm{L};\mathbb{R})$ and $L_{\mathrm{s}/\mathrm{a}}^2(\massh,\di\mu_\mathrm{L};\mathbb{R})$ denotes the closed subspace of symmetric/antisymmetric square integrable functions; the two subspaces are orthogonal with respect the real scalar product due to the symmetry of the coefficient $\en{p}^{-1}$ appearing in the differential of the measure. This decomposition corresponds to the decomposition of the complex Hilbert space $L^2(\massh,\di\mu_\mathrm{L})$ into the direct sum of the subspaces of complex functions that satisfy the relation\footnote{Note that this unique decomposition is explicitly given by $\hat{f
}=(\real\hat{f}^\mathrm{s}+i\immag\hat{f}^\mathrm{a})+(\real\hat{f
}^\mathrm{a}+i\immag\hat{f}^\mathrm{s})$.} $\Bar{{\hat{f}}}(\textbf{p})=\hat{f}(-\textbf{p})$ and the relation $\Bar{{\hat{f}}}(\textbf{p})=-\hat{f}(-\textbf{p})$.\\\\ 

A unitarily equivalent way to characterize the real Hilbert space $\RealHil^\infty$~\cite[Prop.~3.1]{KW91},~\cite[Sec.~2]{ArakiHD2} is via $\RealHil_{\mathrm{pos}}^\infty$, that is the completion of the real vector space $C^{\infty}_c(\mathbb{M},\mathbb{R})/\ker (\operatorname{E})$ in the norm induced by the real inner product defined by the symmetric part $\mu^\infty$ of the two-point function $\omega^\infty_2$ of the ground state \cite[Equation $(2.6)$ and $(2.5)$]{ArakiHD2}. On this real Hilbert space we define a complex structure. It acts on the dense subspace $C^{\infty}_c(\mathbb{M},\mathbb{R})/\ker (\operatorname{E})$ as the operator that implements the multiplication by $i$ on the Fourier transform evaluated on the positive mass-hyperboloid (and therefore by $-i$ on the negative one~\cite[Footnote~6]{ArakiHD2})
\begin{equation*}
\left.\widehat{\ComplexStr_{\RealSub{\mathrm{pos}}^\infty }f}\right\vert_{\massh}=i\left.\hat{f}\right\vert_{\massh},\quad f\in\RealHil^\infty_\mathrm{pos}.
\end{equation*}
Note that this operator is a well defined linear map from $\RealHil^\infty_\mathrm{pos}$ to itself and satisfies all the properties of a complex structure. Using this complex structure we turn $\RealHil_{\mathrm{pos}}^\infty$ into a complex Hilbert space $\Hil_\mathrm{pos}^\infty$ with complex inner product implemented by the ground state two-point function $\omega_2^\infty=\mu^\infty+\frac{i}{2}\sigma$. The isomorphism between the complex Hilbert spaces $\Hil^\infty$ and $\Hil_{\mathrm{pos}}^\infty$ is then implemented on the dense subspace $C^\infty_c(\Mink,\mathbb{R})/\ker(\operatorname{E})$ by the map~\eqref{Eq: fourier_transform_on_mh} and we use the same symbol to denote its extension by continuity to $\Hil^\infty$.\\

The real Hilbert space $\RealHil^\infty_\mathrm{pos}$ can be decomposed in the direct sum of the real subspaces of functions symmetric and antisymmetric in the time variable. Denoting by $\RealSub{_\mathrm{pos}}^\infty$ the subspace generated by symmetric test functions in the time variable, the space generated by antisymmetric test functions coincides with $\ComplexStr_{\Hil^\infty_{\mathrm{pos}}}\RealSub{_\mathrm{pos}}^\infty$ so that we can write $\RealHil^\infty_\mathrm{pos}=\RealSub{_\mathrm{pos}}^\infty\oplus\ComplexStr_{\Hil^\infty_{\mathrm{pos}}}\RealSub{_\mathrm{pos}}^\infty$; the relation $\RealSub{_\mathrm{pos}}^\infty\perp\ComplexStr_{\Hil^\infty_{\mathrm{pos}}}\RealSub{_\mathrm{pos}}^\infty$ directly follows from the symmetry of $\mu^\infty$.This decomposition corresponds, via the isomorphism~\eqref{Eq: fourier_transform_on_mh}, to the decomposition of $\RealHil^\infty=\RealSub{}^\infty\oplus \ComplexStr_{\RealHil^\infty}\RealSub{}^\infty$. Indeed, for any real valued function $f_+(x)$ symmetric in the time variable $x^0$ we have
\begin{equation*}
    \left.\Bar{\hat{f}}_+(p_0,\textbf{p})\right\vert_{\massh}=\int f_+(-x^0,\textbf{x})e^{ix^0\en{p}+i\textbf{x}\cdot\textbf{p}}\di^4x=\left.\hat{f}_+(p_0,-\textbf{p})\right\vert_{\massh}.
\end{equation*}
Finally, let us define the notation for the two following subspaces of $\RealSub{\mathrm{pos}}^\infty$ and $\mathcal{H}^\infty_\mathrm{pos}$
\begin{align*}
 \RealSub{\mathrm{pos}}^\infty(\Ocal)&\coloneqq\overline{\left\{f\in C^\infty_c(\Ocal),\;f(x^0,\textbf{x})=f(-x^0,\textbf{x})\right\}}^{\RealHil^\infty_\mathrm{pos}},\\  (\ComplexStr_{\RealHil_\mathrm{pos}^\infty}\RealSub{\mathrm{pos}}^\infty)(\Ocal)&\coloneqq \overline{\left\{f\in C^\infty_c(\Ocal),\;f(x^0,\textbf{x})=-f(-x^0,\textbf{x})\right\}}^{\RealHil^\infty_\mathrm{pos}},\\
    \mathcal{H}^\infty_\mathrm{pos}(\Ocal)&\coloneqq \overline{\left\{f\in C^\infty_c(\Ocal)\right\}}^{\RealHil^\infty_\mathrm{pos}}.
\end{align*}
Note that if a region $\mathcal{O}$ is invariant under time reflections the decomposition $\RealHil^\infty_{\mathrm{pos}}=\RealSub{\mathrm{pos}}^\infty\oplus\ComplexStr_{\RealHil^\infty_{\mathrm{pos}}}\RealSub{\mathrm{pos}}^\infty$ is compatible with the localization of a real subspace $\mathcal{H}^\infty_\mathrm{pos}(\Ocal)$, in the sense that it is always possible to find subspaces $\RealSub{1},\RealSub{2}$ of $\RealSub{\mathrm{pos}}^\infty$ such that $\mathcal{H}^\infty_\mathrm{pos}(\Ocal)=\RealSub{1}\oplus\ComplexStr_{\RealHil^\infty_\mathrm{pos}}\RealSub{2}$ with $\RealSub{1}=\RealSub{\mathrm{pos}}^\infty(\Ocal),\ComplexStr_{\RealHil^\infty_\mathrm{pos}} \RealSub{2}= (\ComplexStr_{\RealHil_\mathrm{pos}^\infty}\RealSub{\mathrm{pos}}^\infty)(\Ocal)$ generated by functions supported in $\mathcal{O}$.\\

Relying on these definitions we now identify the local von Neumann algebras in the thermal representation of the scalar field, in terms of the corresponding real subspaces as described in the previous section. From now on, the one-particle Hamiltonian $h$ of the ground state is identified with the multiplication operator $\en{p}$ acting on $\Hil^\infty$ and $\complexconj$ with the usual complex conjugation for complex valued functions. The same symbol is also used to denote the real linear operators that they naturally induce on $\RealHil^\infty$. For simplicity, we also restrict our attention to regions $\mathcal{O}$ that are invariant under time reflections.
\begin{lemma}\label{Lem: identification_thermal}
Let $\Ocal \subset \mathbb{M}$ be an open subset of Minkowski spacetime symmetric under time reflections and consider the associated abstract Weyl $\operatorname{C}^*$-algebra $\mathcal{A}(\Ocal)$. Then, the corresponding local von Neumann algebra in the GNS representation induced by a quasi-free KMS state $\omega^{\beta}$ at inverse temperature $\beta \in (0,+\infty)$, is given by
    \begin{equation*}
    \begin{aligned}
         \mathcal{M}_{\beta}(\Ocal) &= \mathcal{R}_{\mathrm{F}}(\UOcal, \VOcal)\\
         &= \left\{U(u),V(v);u\in \UOcal,v\in\VOcal\right\}'',
    \end{aligned}
    \end{equation*}
    where the subspaces $\mathcal{U}_\Ocal,\mathcal{V}_\Ocal$ are subspaces of $\RealSub{}^\infty\oplus \RealSub{}^\infty$ defined by
    \begin{align*}
    &\UOcal\coloneqq \left\{ u_1\oplus u_2 \in  \RealSub{}^\infty\oplus\RealSub{}^\infty : u_1 = \complexconj \sinh ( \implicitop ) K^\infty f , \,\, u_2 = \cosh ( \implicitop ) K^\infty f , \,\, \mathrm{for} \,\, f \in \RealSub{\mathrm{pos}}^\infty(\Ocal)\right\},\\
    &\VOcal \coloneqq \left\{ v_1\oplus v_2 \in  \RealSub{}^\infty\oplus\RealSub{}^\infty : v_1 = -\ComplexStr_{\RealHil^\infty} \complexconj \sinh ( \implicitop )  K^\infty g , \,\, v_2 = -\ComplexStr_{\RealHil^\infty} \cosh ( \implicitop ) K^\infty g , \,\, \mathrm{for} \,\, g \in (\ComplexStr_{\RealHil_\mathrm{pos}^\infty} \RealSub{\mathrm{pos}}^\infty)(\Ocal)\right\}.
\end{align*}
The algebra acts on the Hilbert space $\mathscr{H}_{\omega^{\beta}}$ given in Equation~\ref{Eq: KMS_fock}.
\end{lemma}
\begin{proof}
We start by showing that $\mathcal{U}_\Ocal,\mathcal{V}_\Ocal$ are subspaces of $\RealSub{}^\infty\oplus\RealSub{}^\infty$. The map $K^\infty$ maps element of $\RealSub{\mathrm{pos}}^\infty$ (respectively $\ComplexStr_{\RealHil_\mathrm{pos}^\infty} \RealSub{\mathrm{pos}}^\infty$) into elements of $\RealSub{}^{\infty}$ (respectively $\beta_{\mathcal{H}^{\infty}}\RealSub{}^{\infty}$). The multiplicative operators defined on $\Hil^\infty$
\begin{equation*}
    \sinh(\implicitop)  = \frac{1}{\sqrt{e^{\beta \omega_{\mathbf{p}}}- 1}}, \qquad \cosh(\implicitop)  = \frac{1}{\sqrt{1 - e^{-\beta \omega_{\mathbf{p}}}}}.
\end{equation*}
are bounded and, since $\en{p}=\en{-p}$, the real linear operator that they induce on $\RealHil^\infty$ leaves the subspaces $\RealSub{}^{\infty},\beta_{\mathcal{H}^{\infty}}\RealSub{}^{\infty}$ invariant. Clearly the same holds true for the real linear operator induced by $\complexconj$ on $\RealHil^\infty$. Since the complex structure $\ComplexStr_{\RealHil^\infty}$ anti-commutes with the conjugate linear operator $\complexconj$ and commutes with the linear operator $\en{p}$, the inclusion of $\mathcal{U}_\Ocal,\mathcal{V}_\Ocal$ is verified. In order to construct the algebra $\mathcal{M}_\beta(\Ocal)$ we observe that, since $\complexconj e^{i\en{p}}= e^{-i\en{p}}\complexconj$, the Fock representation induced by the KMS state is consistently obtained making use of the doubling procedure described in Section~\ref{Sec: RealScalarField}. The Fock space of the representation coincides with $\mathscr{H}_{\omega^{\beta}}$ as given in Equation~\eqref{Eq: KMS_fock}. Then, making use of Equation~\eqref{eq: formaKbeta}, we obtain the real subspace $\RealHil^\beta_\Ocal$ that labels the represented algebra in the Segal formulation
\begin{equation*}
    \mathcal{H}^{\beta}_{\Ocal} = \left\{ u_1 \oplus u_2 \in \mathcal{H}^{\infty} \oplus \mathcal{H}^{\infty} : u_1 = \complexconj \sinh ( \implicitop ) K^\infty f , \,\, u_2 = \cosh ( \implicitop ) K^\infty f , \,\, \mathrm{for} \,\, f \in \mathcal{H}^\infty_\mathrm{pos}(\Ocal) \right\}.
\end{equation*}
Finally, by using Equations~\eqref{eq: chiusuravN},~\eqref{eq: SegalWeyl} we identify the corresponding subspaces $\RealSub{}^\infty\oplus\RealSub{}^\infty$ in the Weyl formulation, concluding the proof.
\end{proof}
\begin{remark}
The real subspaces $\UOcal$ and $\VOcal$ play here the role of the real subspaces $\mathcal{U}(\RealSub{1})$ and $\mathcal{V}(\RealSub{2})$ introduced in Proposition~\ref{prop: GenericSub} associated respectively with the real subspaces $\RealSub{1} = K^\infty\RealSub{\mathrm{pos}}^\infty(\Ocal)$ and $\RealSub{2} = \ComplexStr_{\RealHil^\infty} K^\infty(\ComplexStr_{\RealHil_\mathrm{pos}^\infty} \RealSub{\mathrm{pos}}^\infty)(\Ocal)$ of the ground state representation.
\end{remark}

\subsection{Haag duality for causal diamonds}
Now that the explicit form of the local von Neumann algebras $\mathcal{M}_{\beta}(\Ocal)$ with $\Ocal$ invariant under time reflections has been determined, it is the purpose of this section to prove Theorem~\ref{thm: main}. As it is well known~\cite{ArakiHD2}, for the proof of Haag-Duality certain regularity properties of the bounded open region $\mathcal{O} \subset \mathbb{M}$ play an important role. Therefore, in order to avoid geometrical and topological complications (see the discussion in~\cite[Sec.~$7$]{ArakiHD2}), we always assume $\Ocal \subset \Mink$ to be an open causal diamond. By definition, $\Ocal$ is a globally hyperbolic subset of $\Mink$, namely it has an embedded Cauchy hypersurface $\subMink$. Moreover, being also $\Mink$ globally hyperbolic, this embedded hypersurface can always be extended to a Cauchy hypersurface of the full spacetime $\Cauchy \subset \Mink$. The Cauchy development of $\subMink=\Sigma\cap \Ocal$ in $\Mink$, denoted by $\Cdiamond(\subMink)$, coincides with the causal diamond $\Ocal$, namely $\Cdiamond(\subMink) = \Ocal$. Note that the equality $\Cdiamond(\subMink')=\Ocal$ holds for any other choice $\subMink'=\Ocal\cap\Sigma'$, with $\subMink'$ being a Cauchy hypersurface of $\Ocal$. From now on we assume $\mathcal{O} = C(\subMink) \subset \mathbb{M}$ to be a causal diamond and, for the sake of simplicity, we assume $C(\subMink)$ to have base on a fixed time Cauchy hypersurface $\Cauchy_t$. Namely we consider regions of the form
\begin{equation}\label{eq: C(B)}
    \Cdiamond(\subMink)=\{x\in\Mink\vert(x-y)\cdot(x-y)>0\;\;\text{for any}\;\;y\in\Cauchy_t,y\not\in\subMink\}.
\end{equation}
Let us further assume (if not differently stated) the diamond $\Cdiamond(\subMink)$ to have base $\subMink$ on the time zero Cauchy surface $\Sigma_0$, so that $\Cdiamond(\subMink)$ is invariant under time reflections. In order to apply the results of Section~\ref{sec: General} and prove our main theorem, we need to identify the orthogonals of the real subspaces $\RealSub{_\mathrm{pos}}^\infty(\mathcal{O})$ and $(\ComplexStr_{\RealHil_\mathrm{pos}^\infty} \RealSub{_\mathrm{pos}}^\infty)(\mathcal{O})$. This is a known result~\cite{ArakiHD2, Garbarz_2022} of which we report a streamlined proof for completeness. To achieve it, we make use of the globally hyperbolicity of causal diamonds. This allows to establish an isomorphism between $\RealSub{_\mathrm{pos}}^\infty(\mathcal{O})$ and the spaces of initial data on the Cauchy surface $\subMink \subset \Ocal$. Let us define the two following Hilbert spaces

\begin{definition}
We define the real Hilbert spaces of initial conditions as
\begin{align*}
    \Fphi \coloneqq \overline{(S(\mathbb{R}^3, \mathbb{R}), \langle \cdot , \cdot \rangle_{\varphi})}\\
    \Fpi \coloneqq \overline{(S(\mathbb{R}^3, \mathbb{R}), \langle \cdot , \cdot \rangle_{\pi})},
\end{align*}
where the completion is taken in the topology induced by the corresponding inner products that, for $f, g \in S(\mathbb{R}^3, \mathbb{R})$, are defined by
\begin{align*}
    \langle f , g \rangle_{\varphi} &\coloneqq \langle \omega^{-\frac{1}{2}} f, \omega^{-\frac{1}{2}} g \rangle,\\
    \langle f , g \rangle_{\pi} &\coloneqq \langle \omega^{\frac{1}{2}} f, \omega^{\frac{1}{2}} g \rangle.
\end{align*}
Here, on the right hand side, the inner product is the standard inner product on $L^2(\mathbb{R}^3,\mathbb{R})$, amended by the presence of multiplicative operators $\omega^{\alpha}: S(\mathbb{R}^3, \mathbb{R}) \to S(\mathbb{R}^3, \mathbb{R})$, for $\alpha \in \mathbb{Q}$, defined as
\begin{equation*}
    \omega^{\alpha} f  \coloneqq \mathcal{F}^{-1} \left(\en{p}^{\alpha} \hat{f}\vert_{\massh}(\mathbf{p})\right).
\end{equation*}
\end{definition}
From the given definitions of the inner products the following inclusions 
\begin{equation*}
    \Fpi \subset L^2(\mathbb{R}^3,\mathbb{R}) \subset \Fphi,
\end{equation*}
follow. They are implemented by the continuous maps
\begin{equation*}
    j_1 : \Fpi \to L^2(\mathbb{R}^3,\mathbb{R}) \, , \qquad j_2 : L^2(\mathbb{R}^3,\mathbb{R}) \to \Fphi.
\end{equation*}
where continuity follows from the estimates 
\begin{equation*}
    \| j_1 f \|_{L^2} \leq m^{-\frac{1}{2}} \| f \|_{\Fpi} \, , \qquad \| j_2 f \|_{\Fphi} \leq m^{-\frac{1}{2}} \| f \|_{L^2}.
\end{equation*} 
We recall that the real Hilbert spaces of initial conditions are isomorphic respectively to $\RealSub{_\mathrm{pos}}^\infty$ and $\ComplexStr_{\RealHil_\mathrm{pos}^\infty} \RealSub{_\mathrm{pos}}^\infty$ (see Lemma~\ref{lem: isoinitial}), the isomorphisms being implemented by the maps
\begin{equation*}
    \delta_0 : \RealSub{_\mathrm{pos}}^\infty \to \Fphi\, , \qquad \delta_1 : \ComplexStr_{\RealHil_\mathrm{pos}^\infty} \RealSub{_\mathrm{pos}}^\infty \to \Fpi.
\end{equation*}
Their explicit action on $f \in S(\mathbb{R}^4,\mathbb{R}) \cap \RealSub{_\mathrm{pos}}^\infty$ and $g \in S(\mathbb{R}^4,\mathbb{R}) \cap \ComplexStr_{\RealHil_\mathrm{pos}^\infty} \RealSub{_\mathrm{pos}}^\infty$ is given by
\begin{equation}\label{eq: azionedelta}
    \delta_0 f \coloneqq \mathcal{F}^{-1}\left( \hat{f}\vert_{\massh}(\mathbf{p}) \right) \, , \qquad \delta_1 g \coloneqq \mathcal{F}^{-1}\left( (i \omega_{\mathbf{p}})^{-1}\hat{g}\vert_{\massh}(\mathbf{p}) \right)
\end{equation}
and then extended by continuity to $\RealSub{_\mathrm{pos}}^\infty$ and $\ComplexStr_{\RealHil_\mathrm{pos}^\infty} \RealSub{_\mathrm{pos}}^\infty$. If $h \in C^{\infty}_c(\mathbb{M},\mathbb{R})$ such that $\mathrm{supp}(h) \subset \mathcal{O}$, $\mathcal{O} = C(\subMink)$, then the support properties of the causal propagator (see Lemma~\ref{lem: isoinitial}) imply
\begin{equation*}\label{rem: isomorphism}
\mathrm{supp}\left( \delta_0\left(\RealSub{_\mathrm{pos}}^\infty (\Ocal)\right)\right) \subset \subMink, \qquad \mathrm{supp}\left(\delta_1\left((\ComplexStr_{\RealHil_\mathrm{pos}^\infty} \RealSub{_\mathrm{pos}}^\infty)(\Ocal)\right)\right) \subset \subMink.
\end{equation*}
Therefore we can characterise the real subspaces in terms of initial conditions on a Cauchy surface. We define the following subspaces of $\Fphi$

\begin{definition}[\protect{\cite[Eqs.~$(5.1),(5.2)$]{ArakiHD2}}]
Let $\Ocal = C(\subMink)$ be a time reflection invariant causal diamond. Then, its associated spaces of initial conditions are
\begin{align*}
    F_R(\subMink) &\coloneqq \overline{j_2 L^2(\subMink, \mathbb{R})}^{\| \cdot \|_{\varphi}}\\
    F_I(\subMink) &\coloneqq \beta_{\pi,\varphi} j_1^{-1} \left(L^2(\subMink, \mathbb{R}) \cap \Fpi \right)
\end{align*}
where $F_R(\subMink) , F_I(\subMink) \subset \Fphi$ and the completion is with respect to the topology induced by the scalar product on $\Fphi$. Moreover, we have introduced the following operator
\begin{equation*}
    \beta_{\pi,\varphi} \coloneqq \delta_0 \circ  \ComplexStr_{\RealHil_\mathrm{pos}^\infty} \circ \delta_1^{-1} : L^2(\mathbb{R}^3,\mathbb{R}) \to \Fphi.
\end{equation*}
\end{definition}

 The action of the isomorphism~\ref{eq: azionedelta} together with de definition of the real subspaces $\UOcal$ and $\VOcal$ in Lemma~\ref{Lem: identification_thermal} imply that
\begin{align}
    &\UOcal = \left\{f_1\oplus f_2\in\RealSub{}^\infty\oplus\RealSub{}^\infty : f_1 =\frac{\mathcal{F}(f)(-\textbf{p})}{\sqrt{e^{\beta\en{p}}-1}}, \;f_2=\frac{\mathcal{F}(f)(\textbf{p})}{\sqrt{1-e^{-\beta\en{p}}}}, \; f\in F_R(\subMink) \right\} \label{eq: 1partR},\\
    &\VOcal = \left\{f_1\oplus f_2\in\RealSub{}^\infty\oplus\RealSub{}^\infty : f_1 =\frac{-\mathcal{F}(f)(-\textbf{p})}{\sqrt{e^{\beta\en{p}}-1}},\;f_2=\frac{\mathcal{F}(f)(\textbf{p})}{\sqrt{1-e^{-\beta\en{p}}}},\; f\in F_I(\subMink) \right\}.
\end{align}
To determine the orthogonals $\mathcal{U}_\Ocal^\perp,\mathcal{V}_\Ocal ^\perp$ in $\mathcal{\RealSub{}}^\infty\oplus\RealSub{}^\infty$ we use the following lemma
\begin{lemma}[\protect{\cite[Lem.~$2$]{ArakiHD2},~\cite[Thm.~$4.1$]{Garbarz_2022}}]\label{Lem: review}
The following equalities hold
\begin{align*}
    F_R(\subMink)^\perp &=  F_I(\subMink^c),\\
    F_I(\subMink)^\perp &=  F_R(\subMink^c),
\end{align*}
where $\perp$ denotes the orthogonal in $\mathscr{F}_\varphi$.
\end{lemma}
From this it follows
\begin{proposition}\label{prop: ortFRI}
Given the subspaces $\UOcal$ and $\VOcal$, their orthogonal in $\RealSub{}^\infty \oplus \RealSub{}^\infty$ are
\begin{align*}
    \UOcal^{\perp} &= \mathcal{V}_{\Ocal'} \oplus \Tilde{\mathcal{V}}_{\mathbb{M}}\\
    \VOcal^{\perp} &= \mathcal{U}_{\Ocal'} \oplus \Tilde{\mathcal{U}}_{\mathbb{M}}
\end{align*}
where
\begin{align*}
    \Tilde{\mathcal{U}}_{\mathbb{M}} &\coloneqq \left\{ u_1\oplus u_2 \in  \RealSub{}^\infty\oplus\RealSub{}^\infty : u_1 =  \cosh ( \implicitop )K^\infty f , \,\, u_2 =\complexconj \sinh ( \implicitop )  K^\infty f , \,\, \mathrm{for} \,\, f \in \RealSub{\mathrm{pos}}^\infty\right\},\\
    \Tilde{\mathcal{V}}_{\mathbb{M}} &\coloneqq \left\{v_1\oplus v_2 \in  \RealSub{}^\infty\oplus\RealSub{}^\infty : v_1 =  -\ComplexStr_{\RealHil^\infty} \cosh ( \implicitop )  K^\infty g , \,\, v_2 = -\ComplexStr_{\RealHil^\infty} \complexconj \sinh ( \implicitop ) K^\infty g , \,\, \mathrm{for} \,\, g \in (\ComplexStr_{\RealHil_\mathrm{pos}^\infty} \RealSub{\mathrm{pos}}^\infty) \right\}.\\
   \end{align*}
\end{proposition}
\begin{proof}
    The statement is proved by applying Lemma~\ref{Lem: review} and Proposition~\ref{prop: GenericSub}, where we identify $\RealSub{1} = K^\infty\RealSub{\mathrm{pos}}^\infty(\Ocal)$ and $\RealSub{2} = \ComplexStr_{\RealHil^\infty} K^\infty(\ComplexStr_{\RealHil_\mathrm{pos}^\infty} \RealSub{\mathrm{pos}}^\infty)(\Ocal)$.
\end{proof}
We finally prove the main result of this section 
\begin{theorem}
Let $\mathcal{O} \subset \mathbb{M}$ be a open causal diamond with base $\subMink\subset\Cauchy_t$ (not necessarily invariant under time reflections) on Minkowski spacetime $\mathbb{M}$ and $\mathcal{A}(\Ocal)$ the corresponding abstract Weyl $\operatorname{C}^*$-algebra of a free real massive scalar field. Consider $\omega^{\beta}$, for $0< \beta < \infty$, the quasi-free KMS state with respect to the free dynamics on the quasi-local algebra $\mathcal{A}(\mathbb{M})$ and let $\mathcal{M}_\beta(\Ocal)$ be the represented algebra in the GNS representation induced by the KMS state
\begin{equation*}
    \mathcal{M}_\beta(\Ocal) \coloneqq \pi_{\omega^{\beta}}(\mathcal{A}(\Ocal))''.
\end{equation*}
Then, generalised Haag duality holds
\begin{equation*}
    \mathcal{M}_\beta(\Ocal)' = \mathcal{M_\beta}(\Ocal') \vee J\mathcal{M}_{\beta}(\mathbb{M})J.
\end{equation*}
\end{theorem}
\begin{proof}
We begin by assuming the diamond $\Ocal$ to be invariant under time reflections and we conclude deriving the general case by exploiting the covariance of the net. As proven in Lemma~\ref{Lem: identification_thermal}
\begin{equation*}
    \mathcal{M}_{\beta}(\Ocal) = \mathcal{R}_{\mathrm{F}}(\UOcal, \VOcal).
\end{equation*}
The conditions $ K^\infty\RealSub{\mathrm{pos}}^\infty(\Ocal)^\perp \wedge \ComplexStr_{\RealHil^\infty} K^\infty(\ComplexStr_{\RealHil_\mathrm{pos}^\infty} \RealSub{\mathrm{pos}}^\infty)(\Ocal)=\{0\}$ and $ K^\infty\RealSub{\mathrm{pos}}^\infty(\Ocal) \wedge \ComplexStr_{\RealHil^\infty} K^\infty(\ComplexStr_{\RealHil_\mathrm{pos}^\infty} \RealSub{\mathrm{pos}}^\infty)(\Ocal)^\perp=\{0\}$ hold as a direct consequence of Lemma~\ref{Lem: review}. Those subspaces are actually known to be in generic position, see~\cite[Appendix~D]{Garbarz_2022}. Then, by Theorem~\ref{Thm: GS} and Proposition~\ref{prop: ortFRI} we have that
\begin{equation*}
    \mathcal{R}_{\mathrm{F}}(\UOcal, \VOcal)' = \mathcal{R}_{\mathrm{F}}(\mathcal{U}_{\Ocal'}, \mathcal{V}_{\Ocal'}) \vee \mathcal{R}_{\mathrm{F}}(\Tilde{\mathcal{U}}_{\mathbb{M}}, \Tilde{\mathcal{V}}_{\mathbb{M}})
\end{equation*}
and equivalently, as proven in~Proposition~\ref{prop: modcomm} 
\begin{align*}
    \mathcal{R}_{\mathrm{F}}(\Tilde{\mathcal{U}}_{\mathbb{M}}, \Tilde{\mathcal{V}}_{\mathbb{M}}) &= \mathcal{R}_\mathrm{F}(\mathcal{U}_\mathbb{M},\mathcal{V}_\mathbb{M})'\\
    &= J \mathcal{R}_{\mathrm{F}}(\mathcal{U}_{\mathbb{M}}, \mathcal{V}_{\mathbb{M}}) J.
\end{align*}
where $J$ is the modular conjugation associated with the pair $(\mathcal{M}_\beta(\mathbb{M}),\Omega_\beta )$.

Suppose now that $\hat{\Ocal}$ is a generic diamond. It is always possible to write $\hat{\Ocal}=\mathfrak{g}(\Ocal)$ with $\Ocal$ a diamond symmetric under time reflections (i.e. a diamond with base on $\Sigma_0$) and $\mathfrak{g}\in\mathscr{G}$ a time translation. Since the KMS state $\omega^\beta$ is invariant under time translations, it follows from the covariance of the net that it exists a unitary operator $U_\mathfrak{g}$ such that
\begin{equation*}
\mathcal{M}_\beta(\hat{\Ocal})=U_\mathfrak{g}^\dagger\mathcal{M}_\beta(\Ocal)U_\mathfrak{g}.
\end{equation*}
Therefore we have
\begin{align*}
\mathcal{M}_\beta(\hat{\Ocal})'&=U_\mathfrak{g}^\dagger\mathcal{M}_\beta(\Ocal)'U_\mathfrak{g}\\&=U_\mathfrak{g}^\dagger \left(\mathcal{M_\beta}(\Ocal') \vee J\mathcal{M}_{\beta}(\mathbb{M})J\right) U_\mathfrak{g}\\
    &=\mathcal{M_\beta}(\hat\Ocal') \vee J\mathcal{M}_{\beta}(\mathbb{M})J.
\end{align*}
This concludes the proof.
\end{proof}

\begin{remark}
    Note that since the KMS state is not invariant under boost transformations, the covariance of the abstract net cannot be directly employed to prove generalised Haag duality for diamonds obtained by boosting the class of diamonds considered in the previous theorem. This situation is different from the one encountered in the GNS representation induced by a vacuum state. In order to generalise the result to boosted diamonds or to more general Cauchy surfaces, the computations performed in this section need to be adequately adapted.
\end{remark}

\vspace{8mm}
{\bf  Acknowledgments:}
The research of S.G.~is funded by the EPSRC Open Fellowship EP/Y014510/1 and, for part of this work, he also benefited from a Short-Term Scientific Mission funded by COST Action CA21109 – CaLISTA, supported by COST (European Cooperation in Science and Technology). S.G.~is also grateful to the National Group of Mathematical Physics (GNFM-INdAM). L.S.~acknowledges financial support by Italian Ministry of University and Research through the grant PRIN 2022ZE8SC4. L.S.~would also like to thank Silvano Tosi for the support.\\
Both authors benefited from discussions with Bernard Kay, Valter Moretti, and Nicola Pinamonti, to whom they are grateful.
Both authors would also like to thank an anonymous referee who pointed out some inaccuracies in the previous version of the paper that have been corrected in the present one.

\vspace{6mm}
{\bf  Conflict of interest:}
On behalf of all authors, the corresponding author states that there is no conflict of interest.

\appendix
\section{Eckmann-Osterwalder theorem}
We briefly recall a result in~\cite{EckmannOsterwalder1973} that is used in the proof of our main theorem
\begin{theorem}[\protect{\cite[Thm.~$2$]{EckmannOsterwalder1973}}]\label{app: EO}
Let $\RealSub{}$ be a real Hilbert space and consider $\RealSub{1},\RealSub{2} \subset \RealSub{}$ closed subspaces. Denoting by $\RealSub{1}^\perp,\RealSub{2}^\perp$ their orthogonal complements in $\RealSub{}$, and assuming that $\RealSub{1},\RealSub{2}$ are in generic position, then it holds
\begin{equation}\label{eq: EckOst}
    \mathcal{R}_\mathrm{F}(\RealSub{1},\RealSub{2})'= \mathcal{R}_\mathrm{F}(\RealSub{2}^\perp,\RealSub{1}^\perp)
\end{equation}
\end{theorem}

\section{Technical Results}
\begin{lemma}\label{Lem: bounded_inverse}
    Let $\Hil$ be a Hilbert space with scalar product $\ComplexProduct{\cdot}{\cdot}$ and $A$ an operator in $\boundedop{\Hil}$ with bounded inverse $A^{-1}$. Let $V\subset\Hil$ be a closed linear subspace and $V^\perp$ its orthogonal. Then the following relation holds
    \begin{equation}
        (AV)^\perp=(A^*)^{-1}V^\perp.
    \end{equation}
\end{lemma}
\begin{proof}
We have
    \begin{equation*}
        v\in(AV)^\perp\iff\ComplexProduct{v}{Aw}=0, \;\, \forall w\in V\iff \ComplexProduct{A^*v}{w}=0, \;\, \forall w\in V\iff A^*v\in V^\perp\iff v\in(A^*)^{-1}V^\perp,
    \end{equation*}
where in the last implication we used the hypothesis to conclude that $(A^*)^{-1}=(A^{-1})^*$ exists as a bounded operator. Note that $(A^*)^{-1}V^\perp$ is a closed set by the open mapping theorem. 
\end{proof}

\begin{lemma}[\protect{\cite[App.~C]{Garbarz_2022}}]\label{lem: isoinitial}
The real subspaces $\RealSub{_\mathrm{pos}}^\infty$ and $\ComplexStr_{\RealHil_\mathrm{pos}^\infty} \RealSub{_\mathrm{pos}}^\infty$ are isomorphic to the real Hilbert spaces of initial conditions. Namely
\begin{equation*}
    \RealSub{_\mathrm{pos}}^\infty \cong \Fphi \, , \qquad \ComplexStr_{\RealHil_\mathrm{pos}^\infty} \RealSub{_\mathrm{pos}}^\infty \cong \Fpi.
\end{equation*}
\end{lemma}
\begin{proof}
We prove this statement by constructing explicitly the isomorphism. Let us define 
\begin{equation*}
    \delta_0 : \RealSub{_\mathrm{pos}}^\infty \to \Fphi\, , \qquad \delta_1 : \ComplexStr_{\RealHil_\mathrm{pos}^\infty} \RealSub{_\mathrm{pos}}^\infty \to \Fpi,
\end{equation*}
with explicit action on $f \in S(\mathbb{R}^4,\mathbb{R})  \cap \RealSub{_\mathrm{pos}}^\infty$ and $g \in S(\mathbb{R}^4,\mathbb{R}) \cap \ComplexStr_{\RealHil_\mathrm{pos}^\infty} \RealSub{_\mathrm{pos}}^\infty$ given by
\begin{equation*}
    \delta_0 f \coloneqq \mathcal{F}^{-1}\left( \hat{f}\vert_{\massh}(\mathbf{p}) \right) \, , \qquad \delta_1 g \coloneqq \mathcal{F}^{-1}\left( (i \omega_{\mathbf{p}})^{-1}\hat{g}\vert_{\massh}(\mathbf{p}) \right).
\end{equation*}
We start proving that they are isometries of Hilbert spaces and as such their definition can be extended over all of $\RealSub{_\mathrm{pos}}^\infty$ and $\ComplexStr_{\RealHil_\mathrm{pos}^\infty} \RealSub{_\mathrm{pos}}^\infty$.
This follows from the following computation for $f$
\begin{align*}
    \| \delta_0 f \|_{\varphi}^2 &= (\omega^{-\frac{1}{2}} \delta_0 f, \omega^{-\frac{1}{2}} \delta_0 f)\\
    &= \int \frac{\di^3 \mathbf{p}}{\omega_{\mathbf{p}}} \left| \hat{f}\vert_{\massh}(\mathbf{p})\right|^2 = \| f \|_{\mathcal{H}^{\infty}}^2
\end{align*}
and similarly for $g$
\begin{align*}
    \| \delta_1 g \|_{\varphi}^2 &= (\omega^{\frac{1}{2}} \delta_1 g, \omega^{\frac{1}{2}} \delta_1 g)\\
    &= \int \frac{\di^3 \mathbf{p}}{\omega_{\mathbf{p}}} \left| \hat{g}\vert_{\massh}(\mathbf{p})\right|^2 = \| g \|_{\mathcal{H}^{\infty}}^2.
\end{align*}
As isometries, $\delta_0$ and $\delta_1$ and their extensions are injective. Let us further show that they are surjective maps using the symplectomorphism between the symplectic space of initial data and that of sources for the Klein-Gordon equation. Since the operator $P = \Box - m^2$ is Green hyperbolic, we define the corresponding causal propagator (Pauli-Jordan function) as the operator defined on time compact functions $\mathrm{E}: C^\infty_{\mathrm{tc}}(\mathbb{M},\mathbb{R}) \to C^{\infty}(\mathbb{M},\mathbb{R})$ with the property $\mathrm{supp}(\mathrm{E}(f)) \subset J(\mathrm{supp}(f))$ and such that
\begin{equation*}
    \mathrm{E} \circ P \vert_{C^\infty_{\mathrm{tc}}} = 0 \, , \qquad P \circ \mathrm{E} \vert_{C^\infty_{\mathrm{tc}}} = 0.
\end{equation*}
For any $h \in C^\infty_{c}(\mathbb{M},\mathbb{R})$, we have
\begin{equation*}
    (\mathrm{E} h)(x) = \int \mathrm{E}(x,y) h(y)\di^4 y
\end{equation*}
where, using Schwartz's kernel theorem, we denoted by $\mathrm{E}(x,y)$ the distributional kernel of the causal propagator. Then, denoting by $h_{\pm}(x) = \frac{h(x^0,\mathbf{x}) \pm h(-x^0,\mathbf{x})}{2}$, we can explicitly compute
\begin{align}\nonumber
    (\mathrm{E} h)(0,\mathbf{x}) &= \frac{1}{(2\pi)^4}\int e^{i\mathbf{p} \mathbf{x}}\hat{\mathrm{E}}(p)\hat{h}(p) \di^4 p\\
    &= \frac{1}{(2 \pi)^3}\int \hat{h}_-\vert_{\massh}(\mathbf{p}) \frac{\di^3 \mathbf{p}}{i\omega_{\mathbf{p}}} = (\delta_1 h_-)(\mathbf{x})\label{Eq: appB_lem4_1}
\end{align}
and 
\begin{align}\label{Eq: appB_lem4_2}
    -(\partial_{x^0} (\mathrm{E} h))(0,\mathbf{x}) &= (\delta_0 h_+)(\mathbf{x}).
\end{align}
On the other hand consider a pair of initial conditions $f,g \in C^{\infty}_c(\mathbb{R}^3,\mathbb{R})$, which is a dense subset of $\Fphi,\Fpi$. Respectively, identify $f \in \Fpi$ and $g \in \Fphi$. We associate to this pair the corresponding unique spatially compact solution of the homogeneous Klein-Gordon equation $\phi(x^0,\mathbf{x})$ with $f,g$ its initial conditions on the Cauchy hypersurface at $x^0 = 0$ (see e.g.~\cite{Dimock:1980, BaerGinouxPfaffle}). Then, if we consider $\chi \in C^{\infty}(\mathbb{R},\mathbb{R})$ such that, for some $\epsilon > 0$, $\chi(t) = 0$ if $t < -\epsilon$ and $\chi(t) = 1$ for $t > \epsilon$, we can define the test function
\begin{equation*}
    h(x^0,\mathbf{x}) \coloneqq P (\chi(x^0) \phi(x^0,\mathbf{x}))
\end{equation*}
which, for a different choice of $\chi'$ with the same above hypotheses, determines the same element in the symplectic space $C^\infty_c(\Mink,\mathbb{R})/\ker(\operatorname{E})$. Then, since $P (\chi(x^0) \phi(x^0,\mathbf{x})) = - P ((1 - \chi(x^0)) \phi(x^0,\mathbf{x}))$, it follows that $\operatorname{E} h = \phi$ and therefore $h$ satisfies by construction
\begin{equation*}
    (\mathrm{E} h)(0,\mathbf{x}) = f(\mathbf{x}) \, , \qquad (\partial_{x^0} (\mathrm{E} h))(0,\mathbf{x}) = g(\mathbf{x}).
\end{equation*}
Namely, $f,g$ are the initial conditions for the solution sourced by $h$. Together with Equations~\eqref{Eq: appB_lem4_1},\eqref{Eq: appB_lem4_2} this proves surjectivity.
\end{proof}

\section{Pre-cyclicity in the thermal sector}\label{thm: Reeh}
Let $\Ocal$ be an non-empty open subset of Minkowski spacetime $\Mink$ and $\RealHil_\Ocal^\beta\subset\Hil^\infty\oplus\Hil^\infty$, defined as
\begin{equation*}
    \RealHil_\Ocal^\beta\coloneq\left\{\psi_1\oplus\psi_2\in\Hil^\infty\oplus\Hil^\infty:\;\psi_1=\complexconj\sinh (\implicitop)K^\infty f,\;\psi_2=\cosh (\implicitop)K^\infty f,\;f\in C^\infty_c(\Ocal)\right\},
\end{equation*}
the real subspace that labels the local von Neumann algebra $\mathcal{R}_\mathrm{S}(\RealHil_\Ocal)$ associated with $\Ocal$ in the KMS representation of the massive scalar field. We prove that the closure of $\RealHil_\Ocal^\beta$ is standard 
\begin{equation*}
    \overline{\RealHil_\Ocal^\beta+i\RealHil_\Ocal^\beta}=\Hil^\infty\oplus\Hil^\infty.
\end{equation*}
We follow an argument similar to the one used for the proof of~\cite[Prop.~5.6.]{BCV} (see also the proof in~\cite[Sec.~2.5]{Sthesis}).
As proven in~\cite[App.~A.2]{Kay:1985yx}, the complex span of $\RealHil_\Mink^\beta$ is dense in $\Hil^\infty\oplus\Hil^\infty$
\begin{equation*}
    \overline{\RealHil_\Mink^\beta+i\RealHil_\Mink^\beta}=\Hil^\infty\oplus\Hil^\infty.
\end{equation*}
We say that the size of the support of a test function $f\in C^\infty_c(\Mink)$ is smaller than $\Ocal$, and we write $\mathrm{supp}f\prec\Ocal$, if it exists a non empty $I=I_1\times I_2\times I_3\times I_4$ subset of $\mathbb{R}^4$ such that 
\begin{equation*}
    \alpha_u(\mathrm{supp}f)\subset\Ocal\quad \forall u\in I,\quad \alpha_u(\mathrm{supp}f)=\left\{x\in\Mink:\;x-u\in\mathrm{supp}f\right\}.
\end{equation*}
For a generic vector $\Psi \in\RealHil_\Mink^\beta$
\begin{equation*}
    \Psi=\complexconj\sinh (\implicitop)K^\infty f\oplus\;\cosh (\implicitop)K^\infty f
\end{equation*}
we denote by $\Psi_u$ the corresponding translated vector
\begin{equation*}
    \Psi=\complexconj\sinh (\implicitop)K^\infty f_u\oplus\;\cosh (\implicitop)K^\infty f_u,\quad f_u(x)=f(x-u).
\end{equation*}
Every test function $f\in C^\infty_c(\Mink)$ can be decomposed in a finite sum of test functions whose support has size smaller than $\Ocal$ by considering a sufficiently fine partition of the identity. More in details, it always exists a sufficiently large $N\in\mathbb{N}$ and a set of test functions $\chi_j\in C^\infty_c(\Mink),\; j\in[1,N]$ with $\mathrm{supp}\chi_j\subset\mathrm{supp}f$ and $\mathrm{supp}\chi_i\prec\Ocal$ such that
\begin{equation*}
    f=(1-\sum_{j=1}^N\chi_j)f+\sum_{j=1}^N\chi_jf,
\end{equation*}
and $\mathrm{supp}(1-\sum_{j=1}^N\chi_j)f\prec \Ocal$. Therefore we have
\begin{equation*}
    \overline{\operatorname{span}(\RealHil_\Mink^\prec+i\RealHil_\Mink^\prec})=\Hil^\infty\oplus\Hil^\infty,
\end{equation*}
where 
\begin{equation*}
    \RealHil_\Mink^\prec\coloneq\left\{\psi_1\oplus\psi_2\in\Hil^\infty\oplus\Hil^\infty:\;\psi=\complexconj\sinh (\implicitop)K^\infty f,\;\xi=\cosh (\implicitop)K^\infty f,\;f\in C^\infty_c(\Mink),\; \mathrm{supp}f\prec\Ocal \right\}.
\end{equation*}
Suppose now that a vector $\Psi\in\Hil^\infty\oplus\Hil^\infty$ belongs to $(\RealHil_\Ocal^\beta+i\RealHil_\Ocal^\beta)^\perp$ (where $\perp$ denotes the orthogonal with respect to the complex inner product). We show that $\Psi=0$ by showing that it is orthogonal to every vector in the dense set $\operatorname{span}(\RealHil_\Mink^\prec+i\RealHil_\Mink^\prec)$. 
    If $\Xi=\mathscr{R}(\Xi)+i\mathscr{I}(\Xi)$ is a generic vector in $\RealHil_\Mink^\prec+i\RealHil_\Mink^\prec$, the vectors $(\mathscr{R}(\Xi))_u$ and $i(\mathscr{I}(\Xi))_{u'}$
belong to $\RealHil_\Ocal^\beta+i\RealHil_\Ocal^\beta$ for $u$ and $u'$ in non-empty intervals $I,I'\subset\mathbb{R}^4$. We focus on $(\mathscr{R}(\Xi))_u$ since the argument for $i(\mathscr{I}(\Xi))_{u'}$ is analogous. If $u\in I$ the scalar product $\ComplexProduct{(\mathscr{R}(\Xi))_u}{\Psi}$ vanishes by assumption. Thanks to the presence of the operator $\sinh (\implicitop)$, dominated convergence theorem can be used to show that the function $u\mapsto\ComplexProduct{(\mathscr{R}(\Xi))_u}{\Psi}$ is analytic in a complex tube $\mathbb{R}^4+iV^{\beta/2}$ (see~\cite[Sec.~2.5]{Sthesis}). Here $V^{\beta/2}$ denotes the convex set
\begin{equation*}
    V^{\beta/2}=\left\{x\in\mathbb{R}^4:\; x\in V_+\cap(\frac{\beta}{2}e+V_{-})\right\},
\end{equation*}
where $V_{\pm}=\{x\in\mathbb{R}^4:\;\pm x^0>0,x_\mu x^\mu<0\}$ are the future ($+$) and past ($-$) directed causal cones and $e=(1,0,0,0)$. Analyticity then implies that $\ComplexProduct{(\mathscr{R}(\Xi))_u}{\Psi}$ vanishes for any $u\in\mathbb{R}^4$ and, in particular, for $u=0$. By conjugate linearity it follows that $\Psi$ is orthogonal to $\Xi$ and to the dense set $\operatorname{span}(\RealHil_\Mink^\prec+i\RealHil_\Mink^\prec)$ and so $\Psi=0$. In conclusion, since the \textit{complex linear} subspace $\RealHil_\Ocal^\beta+i\RealHil_\Ocal^\beta$ has trivial orthogonal, it is dense in $\Hil^\infty\oplus\Hil^\infty$. 

\printbibliography
\end{document}